\documentclass[cmp]{svjourmod}

 \pdfoutput=1

\usepackage{amssymb}
\usepackage{bm}
\usepackage{psfrag}
\usepackage{amsmath}
\usepackage{amsfonts}
\usepackage{mathrsfs}
\usepackage{cite}

\newcommand{\p}{\partial}

\newcommand{\dd}{{\rm d}}
\begin{document}

\title{Regularity and temperature of stationary black hole event horizons}
\titlerunning{Regularity and temperature of stationary black hole event horizons}

\author{Raymond A. Hounnonkpe\thanks{Universit\'e d'Abomey-Calavi,  B\'enin and Institut de Math\'ematiques et de Sciences Physiques (IMSP), Porto-Novo, B\'enin
 \email{rhounnonkpe@ymail.com}} and Ettore Minguzzi\thanks{Dipartimento di Matematica, Universit\`a degli Studi di Pisa,  Largo
B. Pontecorvo 5,  I-56127 Pisa, Italy. \email{ettore.minguzzi@unipi.it}}}

\institute{}

\date{}

\maketitle

\begin{abstract}
\noindent Available proofs of the regularity of stationary black hole event horizons rely on certain assumptions on the existence of  sections that imply a $C^1$ differentiability assumption. By using a  quotient bundle approach, we remedy this problem by proving directly that, indeed, under the null energy condition event horizons of stationary black holes are  totally geodesic null hypersurfaces as regular as the metric. Only later, by using this result, we show that the cross-sections, whose existence was postulated in previous works,  indeed exist.
These results hold true under weak causality conditions.
Subsequently,
we prove that under the dominant energy condition stationary  black hole event horizons indeed admit  constant surface gravity, a result that does not require any  non-degeneracy assumption, requirements on existence of cross-sections or  a priori smoothness conditions. We are able to make sense of the angular velocity and of the value (not just sign) of surface gravity as quantities related to the horizon, without the need of assuming  Einstein's vacuum equations and the Killing extension.
Physically, this implies that  under very general conditions every stationary black hole has indeed a constant temperature (the zeroth law of black hole thermodynamics).
\end{abstract}


\section{Stationary black holes}

An important problem in mathematical relativity is that of establishing that the event horizon in stationary black holes is as regular as the metric. Subsequently one can try to show further properties, i.e.\ that surface gravity is constant, or that the null vector field tangent to the event horizon extends as a Killing vector field (not necessarily coincident with that implied by the initial symmetry)
in a neighborhood of the horizon (in the analytic case this is Hawking's rigidity theorem \cite{friedrich99,alexakis10}).

Ideally, such a proof should not rely on any assumption on the smoothness of the event horizon nor on special topological conditions. Horizons are differentiable in the interior of the lightlike geodesic generators and at the endpoints of generators with multiplicity one, but otherwise can be quite non-differentiable \cite{chrusciel98,beem98}. The $C^1$ differentiability condition implies that no generator escapes the horizon, that is, every point stays in the interior of a generator \cite[Cor.\ 3.7]{beem98}. Some further fine properties can  be established \cite{chrusciel01,chrusciel02}, for instance by using the theory of lower-$C^2$ functions \cite{minguzzi14d}.
In any case, it  has been observed that imposing strong differentiability properties, and possibly
even analyticity, on the spacetime manifold, metric, or Cauchy hypersurfaces does not
guarantee that  horizons will be differentiable. Indeed, an example by Budzy\'nski et al.\ \cite{budzynski03} shows that non-differentiable compact Cauchy horizons may still form.

However, if the null energy condition is imposed then compact Cauchy horizons are indeed as regular as the metric \cite{minguzzi14d,larsson14}.
 Thus, similar behavior can be expected for event horizons, that is, they might likely be shown to be as regular as the metric under the null energy condition.

Note that the difficulty of the problem lies in the fact that the horizon is not defined locally, e.g.\ via  properties of the Killing field, but rather set theoretically, as the boundary of the past of future infinity. Naturally, it is hard to put one's hand over such evanishing  objects, so technicalities met in works such as   \cite{minguzzi14d,larsson14} devoted to the compact case should be expected. Fortunately, the problem can largely be reduced to an application of the  compact case, and also the existence of a Killing symmetry is of help.
 Techniques for reducing the analysis of non-compact stationary black hole event horizons to the compact  case are not novel.
The idea was introduced in \cite{friedrich99}, see also \cite{moncrief08}, (under smoothness assumptions on the horizon) but, as we shall see, our quotient approach will be different, the quotient being with respect to the isometric flow, not with respect to the geodesic flow of the horizon.

For what concerns available results, a first proof of the smoothness of the event horizon was given in \cite[cf.\ Sec.\ Conclusions]{chrusciel01} but details on how to construct certain spacelike sections pushed to the future by the Killing flow were not provided.

More details were given in Chru\'sciel and Costa \cite[Thm. 4.1,4.11]{chrusciel08}, where the authors made use of a certain assumption on the existence of cross-sections to lightlike geodesic reaching $\mathscr{I}^+$ but, unfortunately, as we shall show,  it implies a $C^1$ differentiability assumption on the future of such section.

Actually, a certain unsatisfaction in regards to the imposition of the existence of such cross-sections is also found in the original paper by Chru\'sciel and Costa \cite[p.\ 197]{chrusciel08} where the authors write

\begin{quote}
We find the requirement (1.1) [a type of cross-section assumption] somewhat unnatural, [\,] but we have not
been able to develop a coherent theory without assuming some version of (1.1). [Without imposing it,] it is not clear how to guarantee
the smoothness of [the horizon] and the static-or-axisymmetric alternative.
\end{quote}



In this work we reconsider and solve this classical problem. In short we shall be able to prove the smoothness of the horizon and the existence of certain useful principal bundles without using assumptions on cross-sections. Nevertheless, the existence of cross-sections will be subsequently proved by using the existence of the bundles.


The paper is structured as follows. In Section \ref{cmqbb}, starting from minimal assumptions, we show the existence of a neighborhood $V$ of the event horizon $H^+$ with the structure of a trivial $\mathbb{R}$-bundle,  the horizon having projection $S$.

In Section \ref{moot}, by considering the $\mathbb{Z}$-action, we show that $H^+$ can be compactified into a horizon $\hat H$, provided $S$ is compact. Under the null convergence condition, this allows us to establish the smoothness of the horizon by using previous results for compact horizons.

In Section \ref{cnqwc}, we establish the existence of cross-sections, that is, compact codimension one submanifolds of the horizon intersected precisely once by every generator.

In Section \ref{ccpt}, we consider the quotient of the horizon with respect to the geodesic flow. This identifies a neighborhood of the event horizon with another trivial $\mathbb{R}$-bundle but with a different quotient. In particular, the projection of the event horizon is now a manifold $\tilde S$. In spacetime dimension 4 and in other cases of interest, $S$ and $\tilde S$ are shown to be homeomorphic.

In Section \ref{sign}, we provide a definition of sign of surface gravity sufficiently general to apply to the non-compact horizon case and relate it to conditions of completeness for a lightlike field $n$ tangent to the horizon. The field $n$ turns out to be unique up to rescalings. When the scale is fixed, typically via other properties, the surface gravity is defined.

Finally, in Section \ref{surf}, we establish that under very general conditions, every event horizon admits a constant surface gravity and an angular velocity. Our findings are then compared with previous results in the literature, particularly those relating to the rigidity and black hole uniqueness theorems.

We end this section by introducing some definitions and terminology which are the same of \cite{minguzzi14d,minguzzi18b}. A spacetime
$(M, g)$ is a paracompact, time oriented Lorentzian manifold of dimension $n + 1 \ge 2$.
The signature of the metric is $(-,+,\ldots,+)$.
We  assume that  $M$ is  $C^k$, $4 \le k \le \infty$, or even analytic, and so as it is $C^1$ it has a unique $C^\infty$ compatible structure (Whitney) \cite[Theor.\ 2.9]{hirsch76}.
Thus we could assume that $M$ is smooth without loss of generality.

The metric will be assumed to be $C^3$ but it is likely that the degree can be lowered. We assume at least this degree of differentiability because it was also assumed in  \cite{minguzzi14d} over which results we rely. A Killing field satisfies $k^a{}_{;b;c}=R^a{}_{bcd} k^d$ and so it is as regular as the metric, and similarly is its flow. For shortness, sometimes we shall sloppily use the word {\em smooth} as meaning {\em as much as the regularity of the metric allows} when such a regularity is clear.

 With a curve symbol we might denote the map or the image of the map. The inclusion $\subset$ is reflexive. We refer the reader to \cite{minguzzi18b} for results in causality theory and conventions not explicitly recalled in the present work.
 
Note added in proof: Our new work (arXiv:2506.20004v2) establishes that the null energy condition suffices wherever the dominant energy condition is used in this paper.





\section{Existence of the principal bundles} \label{cmqbb}

Let $k$ be a complete  Killing field and let $\phi_t: M \to M$ be its flow. As it is customary, we denote $\phi_t(S):=\cup_{p\in S} \{\phi_{t}(p) \}$.
\begin{proposition} \label{nndp}
Let $S\subset M$ be a set invariant under the Killing flow. Then $I^-(S)$, $\overline{S}$ and $\p S$ are invariant under the Killing flow. Thus $\p I^-(S)$ is invariant under the Killing flow. Time dual statements hold.
\end{proposition}

\begin{proof}
Let $t\in \mathbb{R}$. Since timelike curves are sent to timelike curves $\phi_t(I^-(S))=I^{-}(\phi_t(S))$, hence $\phi_t(I^-(S))=I^{-}(S)$.

Since $\phi_t$ is a homeomorphism $\p \phi_t(S)=\phi_t (\p S)$ (and $\overline{\phi_t(S)}=\phi_t(\overline{S})$) thus $\p S=\phi_t(\p S)$ (resp.\ $\overline{S}=\phi_t(\overline{S})$).
$\square$ \end{proof}

We denote  with $\phi(S)$ the orbit of the set $S$ as in \cite{friedrich99}, $\phi(S):=\cup_t\phi_t(S)$.

We assume the existence of an acausal connected spacelike  hypersurface $\Sigma_{end}$, possibly with edge, such that $k$ is timelike on it. We assume that $\Sigma_{end}$ is topologically closed, that is, it includes its edge.

The map $\Sigma_{end} \times \mathbb{R} \to M$, $(p,t)\to \phi_t(p)$ is injective. Indeed, if $\phi_t(p)=\phi_s(q)$, $p,q\in \Sigma_{end}$, then $\phi_{t-s}(p)=q$, so if it were $t\ne s$ there would be a timelike curve connecting $\Sigma_{end}$ to itself contradicting acausality. Thus $t=s$ which implies $p=q$.


Defining the set
\[
M_{end}=\phi(\Sigma_{end})
\]
we get that $M_{end}$ is connected and invariant under the Killing flow and that $k$ is timelike on $M_{end}$ (it can be tricky to prove that $M_{end}$ is a closed set but we shall not use this property).


\begin{proposition} \label{vis}
The {\em event horizon} $H:=\p I^-(M_{end})$ is invariant under the Killing flow. The subsets $H\cap I^+(M_{end})$ and $H\cap \overline{I^+(M_{end})}$ are invariant under the Killing flow (and similarly in the time dual case).
\end{proposition}

\begin{proof}
 It follows from Prop.\ \ref{nndp}.  $\square$
\end{proof}

\begin{definition}
We denote $H^+= H\cap I^+(M_{end})$, and call $M_{out}=I^+(M_{end})\cap I^-(M_{end})$ the {\em domain of outer communication}.
\end{definition}

Since $H$ is an achronal boundary it has no edge. As a consequence $H^+$ might have edge but $\textrm{edge} (H^+)\cap H^+=\emptyset$ (for results on the connection between
 $\textrm{edge} (H^+)$, bifurcate horizons and non-degeneracy, see \cite{kay91,racz92b,boyer69}).

Let $T$ be the open set over which $k$ is timelike.
 Then $T$ is  invariant by the flow and it contains $M_{end}$.

  Let $T(p)$ be the connected components of $T$ including $p$ and let $\gamma_p:=\phi(p)$ be the integral curve of $k$ through $p$.
\begin{lemma}
\label{lem1}
Let $p\in T$ then $I^{-}(\gamma_p, T) = T(p)$. In particular, if  $p\in M_{end}$,  then $M_{end} \subset T(p) = I^{-}(\gamma_p, T)\subset I^{-}(\gamma_p)$.
\end{lemma}
\begin{proof}
Consider $\partial I^{-}(\gamma_p, T)$ in the spacetime $(T,g)$ then $\partial I^{-}(\gamma_p, T)$ is achronal. Since it is invariant under the flow of $k$ which is timelike on $T$ we have  $\partial I^{-}(\gamma_p, T)$ $ =\emptyset$ and so  $ I^{-}(\gamma_p, T) = T(p)$. $M_{end}\subset T$ is connected and contains $p$ so it is contained in the connected component of $T$ that contains $p$, namely $T(p)$.
$\square$ \end{proof}

Although there is some freedom in defining $M_{end}$, the resulting set $I^-(M_{end})$ is largely independent of it as the next results show.
The next result is an improvement of \cite[Lemma 3.1]{chrusciel94c}.

\begin{lemma}
Let $U$ be a past set ($I^-(U)\subset U$) invariant under the flow of $k$, and suppose that  $U\cap M_{end}\ne \emptyset$. Then $I^-(M_{end})\subset U$.
\end{lemma}

Of course, a time dual version holds.

\begin{proof}
Let $r\in U\cap M_{end}$, then from lemma \ref{lem1}, $M_{end}  \subset I^{-}(\gamma_r)$. Since $U$ is invariant and a past set, then
$I^{-}(\gamma_r)\subset U$. Hence $ M_{end} \subset U$ and finally
$I^-(M_{end}) \subset U$ since $U$ is a past set.
$\square$ \end{proof}

%

%
%
%

The next result is essentially \cite[Lemma 3.1]{chrusciel94c} but with weaker assumptions.

\begin{lemma} \label{jwx}
Let $E$ be a subset of $M$ invariant under the flow of $k$, then either  $I^-(E)\cap M_{end}=\emptyset$ or $I^-(M_{end}) \subset I^-(E)$.
\end{lemma}

\begin{proof}
It follows from the previous result.
$\square$ \end{proof}


\begin{corollary} \label{boa}
For  $p\in M_{end}$ we have
$I^-(\gamma_p)=I^-(M_{end})$, $I^+(\gamma_p)=I^+(M_{end})$.
\end{corollary}

\begin{proposition} \label{cppg}
The orbits of $k$ in $M_{end}$ cannot be future (past) imprisoned in a compact set.
\end{proposition}

\begin{proof}
We know that $k$ is timelike on $M_{end}$ and $\Sigma_{end}$ is acausal. Suppose there is a future imprisoned orbit starting from $\Sigma_{end}$, then it  accumulates on an imprisoned inextendible causal curve that accumulates on itself \cite{minguzzi07f,minguzzi18b} (it is not necessarily closed).  Such limit curve would still be an integral curve of $k$ hence timelike, which means, by the openness of the chronology relation, that any point of the limit curve  belongs to the chronology violating set. As the latter is open, the original orbit intersects it, and by isometry that would give that $\Sigma_{end}$ intersects the chronology violating set too, contradicting its acausality.
$\square$ \end{proof}





We recall that a {\em future $C^0$ null hypersurface} $H$ is a locally achronal topological
embedded hypersurface such that for every $p\in H$ there is a  (possibly non-unique) future inextendible lightlike geodesic (called {\em generator}) contained in $H$  with past endpoint $p$  \cite{galloway00}. Future  $C^0$ null hypersurface will also be called  {\em past horizons}.

\begin{proposition}
$H$ is a $C^0$ future null hypersurface.
\end{proposition}

\begin{proof}
Let $r\in M_{end}$, so that $H=\p I^-(\gamma_r)$, where $\gamma_r$, $\gamma_r(0)=r$, is the orbit passing through $r$.
Let $p\in H=\p I^-(\gamma_r)$ then we can find $p_n\in I^-(\gamma_r)$, $q_n\in \gamma_r$, such that $p_n\ll q_n$, $p_n\to p$.
We can always redefine $q_n$ so that $q_n=\gamma_r(t_n)$ with $t_n\to \infty$ and $q_n$ escaping every compact set. By the limit curve theorem, there is a
a future inextendible continuous causal curve $\gamma$  in $\overline{I^-(M_{end})}$ starting from $p$. But this continuous causal curve cannot intersect $I^-(M_{end})$ otherwise $p\in I^-(M_{end})$, which proves that the image of $\gamma$ is contained in $H$. As $H$ is achronal, $\gamma$ is an achronal lightlike geodesic.
%
%
$\square$ \end{proof}

Observe that $H^+$ is an open invariant  subset of $H$. By continuity the Killing orbits cannot leave a connected component $C$ of $H^+$ to enter another one, thus the Killing flow sends each component of $H^+$ to itself, $\phi(C)=C$.

A set with the property of the next proposition is called {\em cross-section} in \cite[Sec.\ 4.1]{chrusciel08}. It is used in that paper to prove smoothness or analyticity of the event horizon. It has become a standard assumption in the literature on black holes, see e.g.\ \cite{kunduri13}.
Unfortunately, it is so strong that it is essentially equivalent to demanding $C^1$ differentiability of $H$ from the outset, as it basically imposes that there are no non-differentiability points.
\begin{proposition}
If $H$ admits a compact subset $K$ intersected precisely once by every generator, then $H$ is $C^1$ on $H\backslash \overline{J^-(K)}$.
\end{proposition}


\begin{proof}
From every non-differentiability point $p$ of a $C^0$ future null hypersurface $H$ depart at least two distinct generators \cite{beem98}. Since every generator intersects $K$ it must be $p\in J^-(K)$, thus $H$ is differentiable on the open set $H\backslash \overline{J^-(K)}$. For a horizon differentiability on an open set is equivalent to $C^1$ regularity \cite{beem98}.
$\square$ \end{proof}

The next result is a slight improvement over \cite[Lemma 2.1]{friedrich99}. We note that, save for the property of acausality for $\Sigma_{end}$, we did not impose causality conditions so far.

 We recall that past distinction holds at a point $q$ if for all $p\in M$, $I^-(p)=I^-(q) \Rightarrow  p=q$. It also admits an equivalent formulation in terms of the existence of past distinguishing neighborhoods cf.\ \cite[Sec.\ 6.4]{hawking73} \cite[Rem.\ 3.12]{minguzzi06c}.

\begin{proposition} \label{vob}
 From every points $q \in  \overline{I^-(M_{end})}$ starts a  future inextendible continuous causal curve $\sigma$ entirely contained in $\overline{I^-(M_{end})}$ (if $q\in H$, by achronality,  this is necessarily a generator of $H$ passing from $q$).

No point  $q\in  \overline{I^-(M_{end})}$ where past distinction holds is such that $I^-(q)\supset I^-(M_{end})$.

If past distinction holds at $q\in  \overline{I^-(M_{end})}\cap I^+(M_{end})$ then $k$ does not vanish at $q$.
\end{proposition}

\begin{proof}
Let  $r\in M_{end}$.  We know that $I^-(\gamma_r)=I^-(M_{end})$ and from Prop.\ \ref{cppg} $\gamma_r$ escapes every compact set in the future. As  $q\in  \overline{I^-(M_{end})}$ we can find $t_k\to +\infty$ and $q_k\to q$ such that $(q_k, \gamma_r(t_k))\in I$. By the limit curve theorem there is a future inextendible continuous causal curve $\sigma$ starting from $q$ and entirely contained in $\overline{I^-(M_{end})}$.

Suppose $q$ as in the second statement exists and let $p\in \sigma$, $p\ne q$, be in the curve $\sigma\subset \overline{I^-(M_{end})}$ constructed in the first paragraph. Then $I^-(q)\subset I^-(p)$ and, as $p$ belongs to the closure of the past set $I^-(M_{end})$,  $I^-(p)\subset   I^-(M_{end}) \subset I^-(q)$, a contradiction with past distinction at $q$.

For the second statement, if $k$ vanishes at $q\in \overline{I^-(M_{end})}\cap I^+(M_{end})$, then $\phi_t(q)=q$ for every $t$, hence $E:=\{q\}$ is invariant under the flow. Observe that $I^-(E)\cap M_{end}\ne \emptyset$ thus, by   Lemma \ref{jwx}, we have $I^-(q)\supset I^-(M_{end})$, which, due to the previously proved result, gives a contradiction.
$\square$ \end{proof}

%
%

\begin{corollary} \label{jid}
For $p\in M_{end}$
no point   $q\in \overline{I^-(M_{end})}$  where past distinction holds is such that $I^-(q)\supset \gamma_p$.
\end{corollary}

\begin{proof}
Indeed, as $I^-(\gamma_p)=I^-(M_{end})$ we would have  otherwise $I^-(q)\supset I^-(M_{end})$, which contradicts the previous proposition.
$\square$ \end{proof}

\begin{proposition} \label{bxa}
Let $p\in M_{end}$.
For each compact subset  $K\subset \overline{I^-(M_{end})}$ at which strong causality holds
we can find $p'\in \gamma_p$ such that $K\cap \overline{I^+(p')}=\emptyset$.
\end{proposition}


\begin{proof}
If not the  family of closed subsets of the compact set $K$ given by $\{\overline{I^+(r)}\cap K: r\in \gamma_p\}$ satisfies the finite intersection property, which implies that there is $q$ common to all elements of the family. But then any $q'\gg q$ is such that $I^-(q')\supset \gamma_p$ and hence $I^-(q')\supset I^{-}(M_{end})$.

As shown in Prop.\ \ref{vob}, there is a future inextendible continuous causal curve $\sigma$ starting from $q$ and entirely contained in $\overline{I^-(M_{end})}$.
Let $r\ne q$ be a point in $\sigma$, to the future of $q$.   Let $x\ll q$, then $r\gg x$, but  $r\in \overline{I^-(q')}$ thus by the openness of $I^+$ we can, with a future directed timelike curve, start from $x$ reach a point arbitrarily close to $r$ and finally reach $q'$. As $x$ and $q'$ can be chosen arbitrarily close to $q$, strong causality is violated at $q$, a contradiction.
$\square$ \end{proof}

The previous result can be improved as follows
%

\begin{proposition} \label{bxb}
Let $(M,g)$ be strongly causal on an  invariant set $A\subset \overline{I^-(M_{end})}$. There is an invariant open set $V \supset A$ such that, for every  $q\in M_{end}$ and for every compact set $K\subset V$  there is $q'\in \gamma_q$ such that $K\cap \overline{I^+(q')}=\emptyset$.
\end{proposition}

 This result can be applied with $A$ replaced by $H$, $H^+$,  any connected component of $H^+$, or $\overline{I^-(M_{end})}$ itself. In the following propositions  we shall only consider the $H^+$ case, though generalizations are possible.

\begin{proof}
Suppose that for every $p\in M_{end}$ there is $V(p)$,  invariant open set containing  $A$, with the property that for every compact set  $K\subset V(p)$   there is $p'\in \gamma_p$ such that $K\cap \overline{I^+(p')}=\emptyset$.

Let us first prove that under this assumption the claim of the proposition holds. Indeed, let $V:=V(r)$ for some $r\in M_{end}$.  Let $q\in M_{end}$ and $K\subset V$. By the assumed property there is $r'\in \gamma_r$ such that $K\cap \overline{I^+(r')}=\emptyset$.
We have $r'\in M_{end} \subset I^-(M_{end})=I^-(\gamma_q)$, thus there is $q'\in \gamma_q$ such that $q'\in I^+(r')$ and hence $\overline{I^+(q')}\subset  \overline{I^+(r')}$ which implies $K\cap \overline{I^+(q')}=\emptyset$.

Let us prove the assumption  in the first paragraph of the proof.
Let $p\in M_{end}$ and let  $q\in A$.  Accordingly to Prop.\ \ref{bxa} there is $p'(q) \in \gamma_p$ such that  $q\notin \overline{I^+(p')}$. As this last set is closed we can find an open set $O(q)$ such that $O(q)\cap  \overline{I^+(p')}=\emptyset$. Let us define the invariant open set
\[
V(p)=\cup\{ \phi_t(O(q)), q\in  A, t\in \mathbb{R}\}=\phi(\cup_q O(q)),
\]
 and let $K\subset V(p)$ be compact. There is a finite covering $\{U_i\}$ of $K$, $U_i:=\phi_{t_i}(O(q_i))$, to whose elements $O(q_i)$ correspond points $p'_i\in \gamma_p$ such that $O(q_i)\cap  \overline{I^+(p_i')}=\emptyset$. The last point $p'$  of the finite family $\{\phi_{t_i}(p'_i)\}$ is then such that $K\cap \overline{I^+(p')}=\emptyset$.
$\square$ \end{proof}


\begin{theorem} \label{kkf}
Suppose that $(M,g)$ is strongly causal at $H^+$. Then there is an invariant open neighborhood $W$ of $H^+$ over which $k\ne 0$ and strong causality holds and defining $V:= W\cap I^+(M_{end})$  the action $\phi: \mathbb{R}\times V \to V$, $(t,p) \mapsto \phi_t(p)$, is proper and free.
\end{theorem}

Note that $H^+$ is closed in the topology of $V$.
\begin{proof}

By  \cite[Prop.\ 4.82]{minguzzi18b} the set at which strong causality holds is open, thus there is an open neighborhood $U$ of $H^+$ at which strong causality holds. By Prop.\ \ref{vob} $k\ne 0$ on $H^+$ hence, by continuity, there is an open neighborhood $U'$ of $H^+$ over which $k\ne 0$. The open neighborhood of $H^+$, $R:=\phi(U\cap U')$, is such that $k\ne 0$ on it and strong causality holds on it.  By Prop.\ \ref{bxb} we can find an invariant open  neighborhood $Z$ as in that result.
Let $W=R\cap Z$ and  $V:=W\cap I^+(M_{end})$.

Suppose the action $\phi$ is not free on $V$. Then we can find $p\in V$ and $t\ne 0$ such that $\phi_t(p)=p$, which implies $p=\phi_{-t}(p)$ and so for every $m\in \mathbb{Z}$, $\phi_{mt}(p)=p$. But for $q\in M_{end}$, $p\in I^+(r)$ for some $r\in \gamma_q$, which implies $p\in I^+(\phi_{mt}(r))$ for every $m\in \mathbb{Z}$ and hence $\gamma_q\subset I^-(p)$, in contradiction with the property of Prop.\ \ref{bxb} for $K=\{p\}$  (or Cor.\ \ref{jid}). This shows that the action $\phi$ is free.

Let $r\in M_{end}$ and let $K_1,K_2\subset V$ be compact subsets.  By Cor.\ \ref{boa} for each $p\in K_1$ we can find some $q\in \gamma_r$ such that $p\in I^+(q)$, thus, passing to a finite subcovering, we see that we can choose $x\in \gamma_r$ such that $K_1\subset I^+(x)$. Moreover, we know from Prop.\ \ref{bxb}, and from $K_2\subset Z$ that there is $y\in \gamma_r$ such that $K_2 \cap \overline{I^+(y)}=\emptyset$, so that the same is true for every $y'\ge y$, $y'\in \gamma_r$.
We can choose $y$ so that  $y\gg x$. Let $\tau>0$ be such that $\phi_\tau(x)=y$. Then it cannot be $\phi_t(K_1)\cap K_2\ne \emptyset$ for any $t\ge \tau$.

Indeed, if there were $z\in \phi_t(K_1)\cap K_2$, for some $t\ge \tau$ then $z':=\phi_{-t}(z)\in K_1$ would be such that $z=\phi_t(z')\in K_2$. Now $z'\in I^+(x)$ thus $z\in I^+(\phi_t(x))$, a contradiction with $I^+(y')\cap K_2=\emptyset$ for $y'\ge y$, $y'\in \gamma_r$.

As $t$ such that $\phi_t(K_1)\cap K_2\ne \emptyset$ is upper bounded, so is $s$ such that $K_1\cap \phi_s(K_2)\ne \emptyset$. But since the latter equation is equivalent to $\phi_{-s}(K_1) \cap K_2\ne \emptyset$, we conclude that the $t$ such that $\phi_t(K_1)\cap K_2\ne \emptyset$ are bounded. The map $\phi: (t,x)\mapsto (\phi_t(x),x)$ is continuous, thus $\phi^{-1}(K_2\times K_1)$ is closed, but $\phi^{-1}(K_2\times K_1)\subset I\times K_1$ with $I$ compact interval, which proves  properness.
$\square$ \end{proof}

Unless otherwise specified in the following  $V$ will be a neighborhood of $H^+$ with the properties of Theorem \ref{kkf}.

From the standard result \cite[Thm.\ 3.34]{alexandrino15} \cite[Thm.\ 1.11.4 55]{duistermaat00} \cite[Thm.\ 21.10]{lee13} we obtain (by {\em manifold} we understand Hausdorff manifold)

\begin{corollary} \label{bjs}
Suppose that $(M,g)$ is strongly causal at $H^+$.
The quotient $V/\mathbb{R}$ is a  manifold and so $\pi: V\to V/\mathbb{R}$ is a trivial principal bundle with structure group $(\mathbb{R},+)$. As a consequence,  $H^+$ is diffeomorphic to the product $\mathbb{R}\times S$, with $S:=\pi(H^+)\subset B:=V/\mathbb{R}$, where $S$ is closed subset of $V/\mathbb{R}$, and where the orbits are the whole $\mathbb{R}$-fibers.
\end{corollary}

The smooth function $t: V\to \mathbb{R}$ constructed in the proof and which realizes the trivialization  will be used in the following (it does not need to be a time function).

\begin{proof}
Indeed the bundle, which exists by the mentioned results, as any $\mathbb{R}$-bundle, is trivial. We recall this fact as it allows us to introduce some notation. Let $\{O_i\}$ be a locally finite covering of $B:=V/\mathbb{R}$ such that the bundle trivializes over $O_i$. Let $\rho_i$ be a smooth partition of unity relative to a covering $\{O_i\}$ and let  $t_i:=\pi_1\circ h_i$ with $h_i:\pi^{-1}(O_i)\to \mathbb{R}\times O_i$  local trivializing maps, then with $t:=\sum_i (\rho_i\circ \pi) t_i$ the map $(t,\pi): V\to \mathbb{R}\times B$, trivializes the bundle (observe that $k( t)= \sum_i (\rho_i\circ \pi) k( t_i)= \sum_i \rho_i\circ \pi =1$, thus $k=\p/\p t$ in the trivialization). Thus if $K\subset V$ is a compact subset, then $t(K)$ is compact which means that no orbit of $k$ can be imprisoned in a compact set (note that the function $t$ increases over the Killing orbits but is not necessarily a time function on $V$).

Since the projection is a quotient map and $H^+=\pi^{-1}(S)$ is closed, it follows that $S$ is closed.
$\square$ \end{proof}

\begin{corollary} \label{erjs}
Suppose that $(M,g)$ is strongly causal at $H^+$.
There is a smooth codimension one hypersurface $\Sigma$ in $V$ which intersects exactly once every Killing orbit (hence those belonging to $H^+$) and which is diffeomorphic to $B$. Its intersection $\sigma=\Sigma\cap H^+$  is diffeomorphic to $S$ and hence has the same number of components of $H^+$. By removing it, $H^+\backslash \sigma$ gets twice the original number of components.

No integral curve of $k$ can forward accumulate on some $q\in V$. In particular, the Killing orbits of $k$ on $H^+$ cannot be closed and they forward escape every compact set.
\end{corollary}

We use the word {\em forward} instead of {\em future} because $k$ can be spacelike on $V$.

\begin{proof}
All the results are consequence of   the $\mathbb{R}$-bundle being trivial.
The hypersurface $\Sigma$ is just a level set $t=cost$ from the trivializing diffeomorphism.
$\square$ \end{proof}

\begin{proposition} \label{act}
Let $(M,g)$ be strongly causal at $H^+$. There are sets $W$ and $V$ as in Theorem \ref{kkf}. Moreover, for any  $\tau>0$ the action $\psi:\mathbb{Z}\times V \to V$,  $(n, p)\mapsto \phi_{n\tau}(p)$ is proper and free.
\end{proposition}

\begin{proof}
Properness is immediate from Theorem \ref{kkf}.

Let $r\in M_{end}$
so that, by the construction of $V$, the property of  Prop.\ \ref{bxb}
holds, namely for every $K\subset V$ we can find $r'\in \gamma_r$ such that $K\cap \overline{I^+(r')}=\emptyset$.


Suppose that there is $n\ne 0$ and $p\in V$  such  that $\phi_{n\tau}(p)=p$. Since $I^+(M_{end})=I^+(\gamma_r)$ we know that there is some $z \in \gamma_r$ such that $p\in I^+(z)$ and hence for every  $k\in \mathbb{Z}$, $p\in I^+(\phi_{kn\tau}(z))$. This implies $I^-(p)\supset \phi(z)=\gamma_r$.
But by Prop. \ref{bxb} this is not possible, just set $K=\{p\}$.
$\square$ \end{proof}

\begin{corollary} \label{bjz}
Let $(M,g)$ be strongly causal at $H^+$.
Let $\tau>0$ and consider the action $\psi$ as in Prop.\ \ref{act}.
The quotient $V/\mathbb{Z}$ is a manifold, and there is a  normal covering map $c: V \to V/\mathbb{Z}$. Furthermore, there is a trivial $S^1$-principal bundle  $\hat \pi: V/\mathbb{Z}\to V/\mathbb{R}$ with $\pi=\hat \pi \circ c$.
\end{corollary}
\begin{proof}
The group $\mathbb{Z}$ endowed with the discrete topology is, being countably infinite, a {\em discrete Lie group}. We have already established that its action is proper and free. By \cite[Thm.\ 21.13]{lee13} there exists a normal covering map $c: V \to V/\mathbb{Z}$.

The projection $\tilde \pi$ is constructed by noticing that every neighborhood of $p\in V/\mathbb{Z}$ is diffeomorphic with a neighborhood of point in the fiber of the covering space. Which point is chosen is irrelevant as different choices shall have canonically diffeomorphic neighborhoods. Thus composing the local diffeomorphism with $\pi$ we get $\tilde \pi$ in a neighborhood of $p$. This shows that the bundle does indeed exist.

The last bundle is trivial because, if $s: B\to V$, $B=V/\mathbb{R}$, is a smooth section for $\pi:V\to B$,  then $c\circ s:B\to V/\mathbb{Z}$ is a section for   $\hat \pi: V/\mathbb{Z}\to B$.
$\square$ \end{proof}

We denote $\hat H=H^+/\mathbb{Z}$ so we have similar restrictions (denoted in the same way), namely a covering $c: H^+\to \hat H$, and a trivial $S^1$-principal bundle $\hat \pi: \hat H \to S$.
%
%
%
%
%
%

\section{Proving smoothness} \label{moot}

The next paragraphs and  Thm.\ \ref{vim} are given to provide characterization (b) in Def.\ \ref{bid} of `horizon with compact section'  but can be skipped on first reading.


The {\em Killing hull} of a set $S$ is the union of the orbit segments of the Killing field that start and end in $S$. A set is {\em Killing convex} if it coincides with its Killing hull.


The next result clarifies what it means for the horizon to have compact  space sections. It does not demand a  compact set to intersect all the generators or all the Killing orbits once, as both requests are strong and should rather be deduced from weaker assumptions. Our disconnection assumption does not mention generators (compare \cite[Sec.\ 4.1]{chrusciel08}).

\begin{theorem} \label{vim}
Let $(M,g)$ be strongly causal at $H^+$.
Let $K$ be a compact subset of $H^+$ which is  Killing convex and whose removal doubles the number of open components (e.g.\ if $H^+$ is connected, removal of $K$ disconnects it in two components).
Then every orbit of $k$ in $H^+$ intersects $K$, that is $H^+=\phi(K)$.
As a consequence, the quotients $S=\pi(H^+)=H^+/\mathbb{R}$, $\hat H=c(H^+)=H^+/\mathbb{Z}$ are compact.
\end{theorem}

\begin{proof}
We use the trivialization throughout the proof.
We know that $H^+$ has the same number of components of $S$. Every point of  $H^+\backslash K$ is of three types: those that belong to an orbit that meets $K$ in the backward direction, those  that belong to an orbit that meets $K$ in the forward direction, those that belong to orbits that do not meet $K$. Those of the first type that project on the same component of $S$, say $C$, belong to the same component of $H^+$. Indeed, their connected orbits in $H^+\backslash K$ will all intersect a suitable hypersurface $\sigma_f$ diffeomorphic to $C$, and not intersecting $K$,  (in a level set of $t$). A similar conclusion is reached for those of the second type, whose connected orbits will intersect $\sigma_p$. If there is a point of the third type that projects in $C$, then its orbit will intersect both $\sigma_f$ and $\sigma_p$ which means that $\pi^{-1}(C)$ is a whole single component. In conclusion, we have two components of  $H^+\backslash K$ for each component of $S$, unless there are points of the third type in which case there would be less components. If the components double there are no points of the third type, namely every orbit intersect $K$, and so $S=\pi(K)$ is compact.
$\square$ \end{proof}

The next definition should not be confused with the definition given in \cite[Sec.\ 4.1]{chrusciel08}. Only in the next section we shall prove the existence of a compact subset of $H^+$ intersected by every generator precisely once, a result which is assumed in \cite{chrusciel08} and which is basically a $C^1$ assumption on $H^+$.

\begin{definition} \label{bid}
We say that $H^+$ has {\em compact projection} if the following equivalent properties hold
\begin{itemize}
\item[(a)]  the projection $S$ is compact,
\item[(b)] There is a  Killing convex compact set $K\subset H^+$ whose removal doubles the number of components.
\end{itemize}
\end{definition}

 Similar definitions apply to the connected components of $H^+$.

\begin{proof}[of the equivalence]
(a)$\Rightarrow$ (b). In the trivialization just let $K$ be any level set of $t$.
(b)$\Rightarrow$ (a). This is Theorem \ref{vim}.
$\square$ \end{proof}

\begin{theorem} \label{bos}
Let $(M,g)$ be a  spacetime endowed with a $C^3$ metric that satisfies the null convergence condition. Let $k$ be a complete  Killing vector field.
Let $\Sigma_{end}$ be an acausal  hypersurface, possibly with edge, over which $k$ is timelike.  Let $M_{end}:=\phi(\Sigma_{end})$ and  suppose that  the horizon $H^+:=\p I^-(M_{end})\cap I^+(M_{end})$, is connected, has compact projection (cf.\ Def.\ \ref{bid}) and that strong causality holds on it. If $\theta\ge 0$ on $H^+$ in the sense of support functions (implied by variants of asymptotic flatness condition, see discussion below)
then $H^+$ is a totally geodesic future null hypersurface as regular as the metric (e.g.\ smooth/analytic if the metric is smooth/analytic).  If $H^+$ is not connected the result holds on each connected component with compact projection provided strong causality and the condition $\theta\ge 0$ hold over such component.
\end{theorem}

The generators can escape $H^+$ in the past direction, hence $H^+$ can have edge. The totally geodesic property implies that the expansion and shear vanish and then the Raychaudhuri equation implies that $R(n,n)=0$ where $n$ is a tangent field to the generators.

\begin{remark}
Let us introduce the property:
\begin{itemize}
\item[$\star$] There is a neighborhood $O$ of $H^+$ such that for every compact set $C\subset O$, $C\cap I^{-}(M_{end})\ne \emptyset$, there is a future complete geodesic $\eta \subset \p J^+(C,{M})$ starting from $C$ that intersects $M_{end}$,
\end{itemize}
From the physical point of view $\star$ states that it is possible to leave the region near the horizon at fastest speed without incurring in additional singularities while reaching the safe region at infinity $M_{end}$.

In the context of spacetimes admitting a conformal completion with asymptotic infinity $\mathscr{I}^+$, the  condition $\star$ can be proved under some reasonable causality conditions. For instance, Hawking deduced it from  stronger but  physically motivated conditions on the asymptotic structure, and in particular from the assumption of {\em asymptotic predictability
} (weak cosmic censorship): $\mathscr{I}^+\subset \overline{D^+(S)}$ where $S$ is a partial Cauchy hypersurface and the closure is in the topology of $\bar{M}$. The reader is referred to \cite{hawking73,chrusciel01} for a discussion of the reasonability of $\star$. In the regularity result by Chru\'sciel and Costa \cite[Thm.\ 4.11]{chrusciel08} it is also present though framed again using
 $\mathscr{I}^+$, see also \cite[Thm.\ 4.10]{chrusciel08}. The formulation $\star$ allows typically for  shorter and less technical  presentations.

As Hawking showed, $\star$ allows one to prove that the expansion is positive on the horizon \cite[Lemma 9.2.2]{hawking73}, a fact which ultimately leads to the proof of the second law of black hole thermodynamics which states that the area of black holes is non-decreasing (see \cite{chrusciel01,minguzzi14d} for a proof without smoothness assumptions).
The fact that Hawking's argument on the positivity of $\theta$ can be adapted to the non-smooth case is  non-trivial and was proved in \cite[Theor.\ 4.1]{chrusciel01}. These authors still work in the geometry of spacetime admitting a suitable conformal completion. This requires more assumptions, though the framework is compatible with that adopted in this work, in which the horizon is defined via the boundary of $I^-(M_{end})$.
We adapt the result  \cite[Theor.\ 4.1]{chrusciel01} to the present framework as follows.
\begin{theorem} \label{pdi}
If $\star$ holds true, then $\theta\ge 0$ on $H^+$.
\end{theorem}

A proof can be obtained from the  sketch  of proof in \cite[Thm.\ 22]{minguzzi14d} (see also the original reference \cite[Thm.\ 4.1]{chrusciel01}), with some trivial replacements such as $I^{-}(\mathscr{I}^+, \bar{M}) \to I^-(M_{end})$. Of course, the inequality $\theta\ge 0$ is understood in a  support function sense, as $H^+$ is a priori non-differentiable.

We are ready to prove the main theorem.
It is interesting to observe that the assumptions  coincide with those that guarantee the validity of the second law of black hole thermodynamics so they are pretty reasonable.
\end{remark}

\begin{proof}[of Theorem \ref{bos}]
 Without loss of generality we can assume that $H^+$ is connected, otherwise  apply the following argument to a single component (note that the previous results generalize to this case without difficulty).

Under the quotient of Corollary \ref{bjz} the horizon $H^+$ projects to a compact $C^0$ future null hypersurface $\hat H$ which is topologically a product. Moreover, we have $\theta\ge 0$ on $\hat H$, thus by the result in
\cite[Thm.\ 1.43]{larsson14} \cite[Thm.\ 18]{minguzzi14d} it follows that $\hat H$ and hence $H^+$ is a totally geodesic lightlike hypersurface as regular as the metric (the Cauchy horizon condition in these theorems is used to infer future completeness of the generators from which $\theta  \ge 0$ is inferred, thus for our purposes the Cauchy horizon condition can be replaced by $\theta \ge 0$, indeed  \cite[Thm.\ 13]{minguzzi14d} implies that $\theta=0$, $\mu^s_{ij}=0$ and  \cite[Thm.\ 17]{minguzzi14d} implies that it is as regular as the metric). Notice that although through each point of $H^+$ passes a unique generator, the generators of $H^+$ need not be past-inextendible (in some examples they are past incomplete and  can be extended escaping $H^+$).
$\square$ \end{proof}

\begin{remark}
Summarizing, from asymptotic conditions it is possible to infer condition $\star$, a kind of future completeness condition, from which, in turn, it is possible to obtain, by using the Raychaudhuri equation and the achronality of the horizon, $\theta\ge 0$, and hence the smoothness of the horizon via Theorem \ref{bos}.

One could think of a different strategy for obtaining $\theta\ge 0$.
In presence of a Killing field we could try to compactify the horizon, as done in Cor.\ \ref{bjz}, and apply the dichotomy non-degenerate/degenerate (i.e.\ all  generators are future complete (or all past complete), or  all incomplete) of compact horizons \cite{moncrief83,moncrief08,reiris21,minguzzi21}. By staying in the future complete case, we could infer $\theta\ge 0$ from the achronality of the horizon and then get the smoothness of the horizon with a proof analogous to the above.

The problem is that available proofs of the mentioned dichotomy apply to compact horizons that are already known to be smooth, so one would have first to extend the ribbon argument and other analytical techniques  used in the proof of the dichotomy  to the non-smooth horizon setting, a strategy which is not entirely clear could be successfully pursued given the technical difficulties involved.
\end{remark}

\section{Existence of cross-sections} \label{cnqwc}

In this section our objective is to obtain as much information as possible on the existence of special sections of $H^+$, particularly with reference to the behavior of generators.

\begin{theorem} \label{vmlh}
Suppose past distinction holds on $H^+$.
For every generator $\gamma:  I \to H^+$, there is no  $s\in I$, such that for some $r>0$,  $\phi_r(\gamma(s)) \le \gamma(s)$ (in particular, Killing orbits might repeatedly intersect the same generator but always in the future direction).
\end{theorem}


\begin{proof}
Let $\gamma: I \to H^+$ be a generator, and assume that there is $s\in I$ such that for some $r>0$,  $\phi_r(\gamma(s)) \le \gamma(s)$.  It follows that for every $n\in \mathbb{N}\backslash\{0\}$, $\phi_{nr} (\gamma(s)) \le \gamma(s)$.
Let $q\in M_{end}$ be such that $\gamma(s) \notin I^+(q)$ (see Prop.\ \ref{vob}). Let $\gamma_q=\phi(q)$, $q=\gamma_q(0)$.
We know that there is some $p=\gamma_q(a)$, $a<0$, such that $\gamma(s)\in I^+(p)$, but for $n$ so large that $a+nr>0$ we have $\phi_{nr} (\gamma(s))\in I^+(\phi_{nr}(p))\subset I^+(q)$, and hence $\gamma(s)\in I^+(q)$, a contradiction.
$\square$ \end{proof}


We recall that a {\em cross-section} in the sense of \cite{chrusciel08} is a topological submanifold of $H^+$ that intersects all generators exactly once. We shall indeed prove that cross-sections exist (in \cite{chrusciel08} this was an assumption)  by taking advantage of the existence of the trivial principal bundle established in the previous section. Note that contrary to \cite{chrusciel08} we do not use the existence of cross-sections to obtain smoothness results, the logical order in this work is reversed ($C^1$ regularity is used to obtain cross-sections, see Lemma \ref{bax} below).


The next result uses ideas in \cite[Remark 2.2]{friedrich99}, see also \cite{chrusciel94b}. In that remark they used without mention a causal simplicity assumption at the end of their argument  (also beware that elsewhere they assume smoothness of the horizon, which, in any case, we proved  without using this type of results). We are able to considerably weaken the causality condition.

\begin{lemma} \label{nkg}
   Suppose that strongly causality holds on $H^+$.
Let $r\in M_{end}$ then every Killing orbit of $H^+$ intersects $Y:=\p I^+(r)\cap H^+$.  Moreover, in the forward direction it is bound to enter $I^+(r)$ and to remain in it once entered. Similarly, in the backward direction, it is bound to escape $\overline{I^+(r)}$ and to remain outside it once escaped. Thus the orbit intersects $Y$ in a compact segment (possibly a point).

If past reflectivity holds in an open set containing  $\overline{I^-(M_{end})} \cap I^+(M_{end})$ then the intersection is just one single point.
\end{lemma}

The curious and interesting role of past reflectivity in black hole physics has been  pointed out in \cite{minguzzi20}. We recall that it is implied by causal continuity and hence by global hyperbolicity. It is also implied by the existence of a continuous Lorentzian distance function \cite{beem77} \cite[Prop.\ 5.2]{minguzzi18b}. Finally, it is implied by the existence of a {\em timelike} Killing field which is complete in the past direction \cite{clarke88} \cite[Thm.\ 4.10]{minguzzi18b}.  As a consequence, past reflectivity holds on $M_{end}$, however, a priori it might not extend up to the horizon.

 It is worth recalling that past reflectivity and past distinction imply stable causality  \cite[Thm.\ 4.111]{minguzzi18b} and that we are imposing strong causality in a neighborhood of $H^+$. So imposing the property of past reflectivity implies the existence of a time function in a neighborhood of the horizon.


\begin{proof}
Let $p\in H^+$ then as $H^+\subset I^+(\phi(r))$,  we can find $\tau>0$ such that $p\in I^+(\phi_{-\tau}(r))$ (remember that $k$ is timelike in $M_{end}$ but not necessarily near the horizon),  and
hence $\phi_{\tau}(p)\in I^+(r)$. Moreover, by Prop.\ \ref{bxa} we can find $s>0$ such that $p\notin \overline{I^+(\phi_s(r))}$ thus $\phi_{-s}(p)\notin \overline{I^+(r)}$. Thus the Killing orbit $x(t)=\phi_t(p)\subset H^+$ intersects $\p I^+(r)$ and  in the forward direction it  enters $I^+(r)$ while in the past direction it  escapes $\overline{I^+(r)}$.

Now observe that if $q\in I^+(r)\cap H^+$ then for $\tau\ge 0$, $q\in I^+(\phi_{-\tau}(r))$  which implies $\phi_\tau(q)\in I^+(r)$, namely following any orbit in the  forward direction we have that once it enters $I^+(r)$ it remains in it.

Similarly, if $q\notin \overline{I^+(r)}\cap H^+$ then for $s\ge 0$, $q\notin  \overline{I^+(\phi_s(r))}$ which implies $\phi_{-s}(q) \notin  \overline{I^+(r)}$, namely following any orbit in the  backward direction  we have that once the
orbit escapes $\overline{I^+(r)}$ it remains outside  it.
This shows that every orbit intersects $Y$ in a compact segment.

Assume past reflectivity. If $q\in \overline{I^+(r)}\cap H^+$ then by past reflectivity $r\in \overline{I^-(q)}$ and for  $\tau> 0$,   as $r\gg \phi_{-\tau}(r)$, $q\in I^+(\phi_{-\tau}(r))$, and by the isometry $\phi_\tau(q)\in I^+(r)$,
namely any orbit can have at most one point in  $\p I^+(r)$  as any subsequent point belongs to $I^+(r)$.
$\square$ \end{proof}

\begin{lemma} \label{bhd}
  Suppose that strongly causality holds on $H^+$ and $S=\pi(H^+)$ is compact.
Let $r\in M_{end}$, then the set $Y:=\p I^+(r)\cap H^+$ is non-empty and compact.
 If past reflectivity holds in  a neighborhood of $\overline{I^-(M_{end})} \cap I^+(M_{end})$, then $Y$ and $S$ are homeomorphic.
\end{lemma}

\begin{proof}
It is non-empty because, as shown above,  every orbit intersects it. Let us consider the splitting $H^+\cong \mathbb{R}\times S$. If it is non-compact we can find, without loss of generality, a sequence $p_n=(t_n,z_n)\in Y$, $z_n\in S$, $z_n\to z\in S$, and $t_n\to +\infty$ or $t_n\to -\infty$.

Consider the former case. The orbit projecting to $z$ enters $I^+(r)$, thus there is some $q\in I^+(r)$ in the orbit projecting on $z$. As $I^+(r)$ is open the orbits projecting to $z_n$ intersect $I^+(r)$ at points $q_n$ for which  $t(q_n)=t(q)$.
 This implies that  $t_n<t(q_n)\le t(q)$ for $n$ sufficiently large,  a contradiction.

Consider the latter case. The orbit projecting to $z$ escapes $\overline{I^+(r)}$, thus there is some $q\notin \overline{I^+(r)}$ in the orbit projecting to $z$. As $\overline{I^+(r)}$ is closed the orbits projecting to $z_n$ intersect $H^+\backslash \overline{I^+(r)}$ at points $q_n$ for which    $t(q_n)=t(q)$.
This implies that  $t_n>t(q_n)=t(q)$ for $n$ sufficiently large,  a contradiction.

The contradiction proves that $Y$ is compact.

 The bundle projection $\pi$ is continuous and under past reflectivity its restriction to $Y$ sends bijectively $Y$ to $S$. But $Y$ is compact and $S$ is Hausdorff thus such restriction is a homeomorphism.
$\square$ \end{proof}

\begin{lemma} \label{bax}
  Suppose that strongly causality holds on $H^+$ and $S=\pi(H^+)$ is compact.
Suppose that $H^+$ is $C^1$ (or assume the hypothesis of Theorem \ref{bos}).
Let $\gamma$ be a generator of $H^+$. Then with reference to the splitting of Cor.\ \ref{bjs}, $\gamma$ intersects all level sets $t=const.$ More precisely, over $\gamma$ we have $t\to +\infty$  in the future direction, and $t\to -\infty$ in the past direction.
\end{lemma}
Notice that it can intersect them more than once.
\begin{proof}
Let $\gamma$ be a generator and let $p$ be one of its points.
Let us affine parametrize it so that $p=\gamma(0)$. There is an inextendible lightlike geodesic $\eta$ such that $\eta(0)=\gamma(0)$, $\dot \eta(0)=\dot \gamma(0)$. The domain of definition $I$ of $\gamma$ is the largest interval in the domain $J$ of $\eta$ that contains $0$ and such that $\eta(I)\subset H^+$. As $\gamma$ is future inextendible the interval is of the form $I=(a,b)$ or $I=[a,b)$ (the constants $a,b$ can take infinite value). But the latter possibility is  excluded because $H^+$ is $C^1$ and so generators have no past endpoint in it, i.e.\ every point of $H^+$ belongs to the interior of a generator \cite{beem98} (still $\eta$ can escape $H^+$ from $\textrm{edge} ( H^+)$, we recall that $\textrm{edge} (H^+)\cap H^+=\emptyset$). In other words the $C^1$ property of $H$ implies that $\gamma$ is past inextendible in $H^+$.


Let us redefine the section that trivializes the bundle $V$ so that $t(p)=0$. We know from Prop.\ \ref{bxa} that there is $r\in M_{end}$ such that $\{t=0\}\cap \overline{I^+(r)}=\emptyset$ (recall that the level sets of $t$ are diffeomorphic to $S$, hence compact). By Lemma \ref{bhd}, $K=\p I^+(r)\cap H^+$ is compact. Let $\tau=\max_{K} t<\infty$, then by Lemma \ref{nkg}  $\{x\in H^+:t(x)> \tau\}\subset {I^+(r)}$.
Observe that for $s\le 0$ it cannot be $\gamma(s)\in I^+(r)$ otherwise $p\in  I^+(r)$, a contradiction. Thus $\gamma((a,0])\subset t^{-1}( (-\infty, \tau])$.   Since strong causality implies non-partial imprisonment, for any $c \in (-\infty, \tau]$,  the curve $\gamma_{(a,0]}$
is bound to escape the compact set  $t^{-1}([c,\tau])$ in the past direction,
due to the past inextendibility of $\gamma_{(a,0]}$ in $H^+$ consequence of the $C^1$ assumption.
Thus we  have $t(\gamma(s))\to -\infty$ for $s\to a$, which proves that it intersects any level set $t=c<0$.


Next observe that there is $q\in \gamma_r=\phi(r)$, $q\le r$, such that $p\in I^+(q)$. The set $K'=\p I^+(q)\cap H^+$ is compact, let $\tau'=\min_{K'} t<\infty$. Again by  Lemma \ref{nkg}  $\{x\in H^+: t(x) < \tau'\}\cap {I^+(q)}=\emptyset$, thus $\gamma([0,b))\subset t^{-1}([\tau',+\infty))$. Again by non-partial imprisonment (implied by strong causality) and future inextendibility of $\gamma$, we have  $t(\gamma(s))\to +\infty$ for $s\to b$, which proves that $\gamma_{[0,b)}$ intersects any level set $t=c'>0$.
$\square$ \end{proof}


\begin{lemma} \label{nog}
 Suppose that strongly causality holds on $H^+$ and $S=\pi(H^+)$ is compact.
Suppose that $H^+$ is $C^1$ (or assume the hypothesis of Theorem \ref{bos}).
The sets of type $Y:=\p I^+(r)\cap H^+$, $r\in M_{end}$, are intersected by each generator exactly once.  In the future direction the generator enters $I^+(r)$ and remains in it and in the past direction it enters $M\backslash \overline{I^+(r)}$ and remains in it. Thus cross-sections in the sense of \cite{chrusciel08} exist.
\end{lemma}

We have shown above that under past-reflectivity $Y$ also intersects the Killing orbits exactly once.

\begin{proof}
 We keep making use of the splitting of Cor.\ \ref{bjs}.
By Lemma \ref{nkg} and compactness of $Y=\p I^+(r)\cap H^+$, $I^+(r)$ contains a set $(\tau, +\infty)\times S$ for some $\tau$, thus by Lemma  \ref{bax} $\gamma$ enters $I^+(r)$ in the future direction. Similarly, by Lemma \ref{nkg} and compactness of $\p I^+(r)\cap H^+$, $\overline{I^+(r)}$ has empty intersection with a set of the form  $(-\infty, \tau')\times S$ for some $\tau'$, thus by Lemma  \ref{bax} $\gamma$ escapes $\overline{I^+(r)}$ in the past direction. Thus $\gamma$ intersects $Y$.

Clearly, if the generator enters $I^+(r)$ it remains there in the future direction. The same is true if it escapes $\overline{I^+(r)}$ in the past direction (this follows, from the more general result on the transitivity of the relation  $D_f^+$ given by $D_f^+ = \{ (x,y): y \in \overline{I^+(x)}\}$, c.f.\ \cite{dowker00} \cite[Thm.\ 3.3]{minguzzi07b}), thus the generator intersects $Y$ in a connected compact subset of its domain of definition.

Now, suppose that there are two points $x,y$, $x \le y$, in the generator such that $x\in  \overline{I^+(r)}$. Since $D_f$ is transitive  $y\in \overline{I^+(r)}$, and our goal is to establish that $y\in I^+(r)$ as that would prove that the generator can intersect $Y$ in at most one point. By the limit curve theorem there is a continuous causal curve $\sigma^x$ with future endpoint $x$ such that $\sigma^x\subset  \overline{I^+(r)}$  which is either  past inextendible or connects $r$ to $x$. Let $\gamma$ be the segment of generator between $x$ and $y$. If the causal curve composition of $\sigma^x$ and $\gamma$ is not achronal, then some point of $z\in \sigma^x$ belongs to $I^-(y)$ and hence $y\in I^+(r)$, as we desired to prove. The other case is not realized because if achronality of the said composition holds then such curve is the prolongation of the generator passing  through $y$ in the past direction (in principle it might leave $H^+$ after reaching its edge). But we know by Lemma \ref{bax} that, as long as it stays in $H^+$, the function $t$ of the trivialization of $H^+$ will go to $-\infty$. This gives a contradiction because as $\sigma^x\subset  \overline{I^+(r)}$ that would imply that  $\overline{I^+(r)}$ intersects $(-\infty, \tau')\times S$ in contradiction with what we established above. This concludes the proof.
$\square$ \end{proof}


\section{Geodesic flow bundle and  unambigous topology of the quotients} \label{ccpt}

Let $n$ be a future-directed lightlike vector field on the horizon, tangent to it. By using the time orientation, and hence the existence of a future-directed timelike vector field, $H^+$ can be endowed with a Riemannian metric \cite{nomizu61}, and hence with a complete Riemannian metric, from which it follows that $n$ can be suitably rescaled to have norm less than one so as to be complete. We denote with $\varphi_t: H^+\to H^+$ the flow of $n$ for a choice of complete field.\\

\begin{theorem} \label{vobn}
 Suppose that strongly causality holds on $H^+$ and $S=\pi(H^+)$ is compact.
Suppose that $H^+$ is $C^k$, $k\ge 3$,   (or assume the hypothesis of Theorem \ref{bos}).
The flow $\varphi$ of $n$ is proper and free.  There exists a compact quotient manifold $\tilde S$ and  $H^+$ is a (trivial) $(\mathbb{R},+)$-bundle $\tilde \pi: H^+ \to \tilde S$. Thus $H^+$ is diffeomorphic to $\mathbb{R}\times \tilde{S}$. The manifold $\tilde S$ is homeomorphic to the sets of the form $Y:=\p I^+(r)\cap H^+$,  for any $r\in M_{end}$.
 If  past reflectivity holds in a neighborhood of  $\overline{I^-(M_{end})} \cap I^+(M_{end})$, then  $\tilde S$ is homeomorphic to $S$ (hence they are diffeomorphic if  $n-1\le 3$).
\end{theorem}

This is a different bundle to that found previously i.e.\ with respect to the geodesic flow $\varphi$, instead of the Killing flow $\phi$.

\begin{proof}
Since $n$ does not vanish  and strong causality holds on $H^+$, there are no $t\neq 0$ and $p\in H^+$ such that $\varphi_t(p) = p$. So the action  is free.
By \cite[Prop.\ 21.5(c)]{lee13} we need to show that for every compact set $K\subset H^+$, $\{t: \varphi_t(K)\cap K\ne \emptyset\}$ is compact. Closure follows noticing that if $t_i$ is a sequence belonging to the set that accumulates to $t$, there are $p_i,q_i \in K$ such that $\varphi_{t_i}(p_i)=q_i$. Passing to a subsequence we can assume that $p_i$ and $q_i$ are convergent to $p,q\in K$, thus, by continuity, $\varphi_t(p)=q$ which proves that $t$ belongs to the set.

As for boundedness, let us consider a similar sequence $p_i,q_i=\varphi_{t_i}(p_i) \in K$,  $p_i\to p$, $q_i\to q$, $p,q\in K$, where this time $t_i$ is an unbounded sequence. We want to find a contradiction. By passing to a subsequence or inverting the roles of some $p_i$ with $q_i$ we can assume that $t_i$  has limit $+\infty$.
 There is $r\in M_{end}$ such that $K\cap \overline{I^+(r)} = \emptyset$. By Lemma \ref{nog} the $n$-parametrized generator $\eta$ starting from $p$ will enter $I^+(r)$ eventually, $\eta(\tau)\in I^+(r)$ and stay there for $t>\tau$.  But the $n$-parametrized geodesic segments $\gamma_i$ connecting $p_i$  to $q_i$ converge to $\eta$ and so, as $t_i\to \infty$, must enter $I^+(r)$ which implies that $q_i\in I^+(r)$ for sufficiently large $i$, a contradiction.   We conclude that the bundle $\tilde \pi: H^+ \to \tilde S$ exists.

By Lemma \ref{nog} we have $\tilde S=\tilde \pi(Y)$, and from Lemma \ref{bhd}  we know that $Y$ is compact, thus $\tilde S$ is compact.
The projection $\tilde \pi$ is continuous so its restriction $\tilde \pi_Y$ is continuous and bijective. But  $Y$ is compact and $\tilde S$ is Hausdorff thus $\tilde \pi_Y$ is a homeomorphism.

The last statement follows from Lemma \ref{bhd}  and from the fact that two differentiable manifolds of dimension less or equal to 3 that are homeomorphic are diffeomorphic.
$\square$ \end{proof}

The following result will not be used but answers a natural question

\begin{proposition} \label{nog2}
Under the assumptions of the previous Lemma, the bundle $\tilde \pi: H^+ \to \tilde S$ admits  a continuous section $\tilde s: \tilde S \to H^+$ with image $Y$. It can be approximated, as much as desired by a smooth section.

Assume past reflectivity.   The bundle $\pi: H^+ \to  S$ admits  a continuous section $s: S \to H^+$ with image $Y$. It can be approximated, as much as desired by a smooth section.
\end{proposition}

\begin{proof}
Both bundles have model fiber $\mathbb{R}$ which is convex. Thus any continuous section can be approximated \cite[Sec.\ 6.7]{steenrod51}.

Since the bundle $\tilde \pi: H^+ \to \tilde S$   is trivial and $Y$ is intersected exactly once by every generator we can use a trivialization $\beta: H^+\to \mathbb{R} \times \tilde S$  to express  $\beta(Y)$ as the graph of a function $h: \tilde S \to \mathbb{R}$. But $Y$ is compact so $\beta(Y)$ and hence the graph of $h$ are compact. A real function over a Hausdorff space having compact graph is continuous, thus $h$ is continuous and hence $\tilde s=\beta^{-1}\circ (h,Id): \tilde S \to H^+$ is continuous and a section with image $Y$.


The case for the Killing  orbits instead of generators is similar.
$\square$ \end{proof}

For the conclusion that $\tilde S$ is homeomorphic to $S$ we are going to show that the assumption of past reflectivity can be dispensed of in most cases of interest.
We recall that two manifolds $M_1$ and $M_2$ are said to be {\em $\mathbb{R}$-diffeomorphic } if $M_1\times \mathbb{R}$ is diffeomorphic to $M_2\times \mathbb{R}$. A classical problem in differential topology is whether two $\mathbb{R}$-diffeomorphic manifolds are also diffeomorphic \cite{hausmann18}.

We have the following result \cite{hausmann18}

\begin{theorem}
Let $M$ and $N$ be two closed manifolds of dimension $k \le 3$, which
are orientable if $k= 3$. Then $N\cong_{\mathbb{R}-\textrm{diff}} M$ implies  $N \cong_{\textrm{diff}} M$.
\end{theorem}
Further results are also available, for instance when one of the two manifolds   is simply connected. Correspondingly the next results could be similarly generalized.

On a $n+1$ dimensional spacetime $\textrm{dim} S=\textrm{dim} \tilde S=n-1$ thus, thanks to the  compactness result for $\tilde S$ (Theorem \ref{vobn}), we have


\begin{theorem}
 Suppose that strongly causality holds on $H^+$ and $S=\pi(H^+)$ is compact.
Suppose that $H^+$ is $C^k$, $k\ge 3$,   (or assume the hypothesis of Theorem \ref{bos}).
Let the dimension of spacetime be 3 or 4, or suppose that it is 5 but $H^+$ is orientable.
Then the Killing flow quotient $S$ is diffeomorphic to the generator flow quotient $\tilde S$.
\end{theorem}

Due to this result, in the physical 4-dimensional spacetime case, there is no ambiguity in speaking of the Euler characteristic of the event horizon projection  even if past reflectivity does not hold. We recall that some arguments due to Hawking \cite{hawking72} \cite[Prop.\ 9.3.2]{hawking73}  imply that in a 4-dimensional spacetime the topology of the black hole event horizon must be  $\mathbb{S}^2$ hence with Euler characteristic different from zero, see also \cite{chrusciel94b,galloway95,jacobson95}.
In Theorem \ref{voz} we just need the latter property.

It is also worth noting that the fact that the topology is $\mathbb{S}^2$ can in most cases be deduced from the following fact without using any assumption on the asymptotic behavior of the spacetime (for a similar argument see \cite{geroch82}). Notations are chosen so that it can be directly applied to our framework. We recall that  the flow of $n$ preserves the induced metric on the horizon  as a consequence of the total geodesic property (implied by Thm.\ \ref{bos})  which therefore passes to a Riemannian metric $\tilde g$ on the quotient $\tilde S$ \cite[Lemma B1]{friedrich99}\cite{moncrief08}\cite[Lemma 7]{minguzzi21}. The Killing field $k$ sends $H^+$ to itself and so it is tangent to it (recall that $H^+$ is smooth, again by Thm.\ \ref{bos}). The Killing flow $\phi$ necessarily sends  generators to generators which implies that it projects to a flow $\tilde \phi$  generated by a Killing field $\tilde k:=\tilde\pi_*(k)$ for $\tilde g$.

\begin{proposition} \label{vnqp}
Let $(\tilde S,\tilde g)$ be a connected closed 2-dimensional Riemannian manifold  admitting a non-trivial $S^1$-isometric action induced by a Killing field $\tilde k$. Then the topology is $\mathbb{S}^2$ (two zeros), $\mathbb{T}^2$ (no zeros), $\mathbb{P}^2$ (one zero) and (Klein bottle) $\mathbb{P}^2\sharp \mathbb{P}^2$ (no zeros).
\end{proposition}

As we shall see in a moment, the non-triviality is related   to what is a posteriori interpreted as the rotation of the black hole.  This `rotational' case is {\em physically} sufficient to determine the topology of the horizon, as one expects that there does not exist any precisely non-rotating black hole in Nature (also because the absorption of just one electron having angular momentum would spoil such a property while it is not expected to alter dramatically the topology of the horizon). The orientability of the horizon can then be deduced from that of the spacetime and that of the Cauchy surface, if there is any, or from other more refined arguments. That leaves only with the options $\mathbb{S}^2$ and $\mathbb{T}^2$ to be discussed which is, to be precise, also the result of Hawking's theorem, the $\mathbb{T}^2$ case being a sort of exceptional case which in every version of the theorem requires a special analysis.

\begin{proof}
One observes that $k$ has isolated zeros or vanishes identically (here one works on a disk at the zero point and uses the fact that if the Killing vector has vanishing covariant derivative at the zero point then it is zero everywhere) and moreover each zero point has index 1 (because the flow sends shortest curves to shortest curves, and preserves small disks  implying that the Killing field is tangent to the disk). Thus the Euler characteristic is larger or equal to 0. Every compact connected surface is homeomorphic to  $\mathbb{S}^2$,  the connected sum of $m$ tori $ \mathbb{T}^2$ or the connected sum of $m$ projective spaces  $\mathbb{P}^2$, $m\ge 1$, of Euler characteristics $2$, $2-2m$, $2-m$, respectively. Thus the only possible cases are those listed.
$\square$ \end{proof}

\section{The (sign of) surface gravity (including the non-compact case)} \label{sign}

In this section we consider a $C^0$ future null hypersurface $N$. For each point $p\in N$ we can find a future-directed inextendible lightlike geodesic $\gamma$  with starting point $p$. Its tangent vector at $p$ is called semi-tangent to $N$ at $p$, see e.g.\ \cite{chrusciel01,minguzzi14d}. Each point has as many semi-tangents as the generators passing from it.  Given a smooth global future-directed timelike vector field $T$, the semi-tangent can be normalized via the condition $g(T,\dot \gamma)=-1$. A choice for the semi-tangent multi-valued map on $N$ will be called a {\em semi-tangent section} $n$. It is understood that for each point $p\in N$ and for each generator passing/starting from $p$ the semi-tangent section contains just one element tangent to the generator and belonging to $T_pM$.

This semi-tangent multi-valued map admits a natural uni-valued restriction over each generator.
Over the points of a  generator $t\mapsto \gamma(t)$ consider the map $t\mapsto n(t)\in T_{\gamma(t)}M$ which provides the tangent to the generator which is coincident with one element of the semi-tangent section.
If the semi-tangents are normalized as above this map is actually $C^1$ due to the $C^2$ regularity   of the geodesic. Moreover, if the semi-tangents $n_k$ have base points $p_k\in N$, $p_k\to p\in N$ and there is a limit $n_k\to n\in T_pM$ (this is always the case up to subsequences due to the introduced normalization), then $n$ belongs to the semi-tangent section \cite[Lemma 3.1(1)]{galloway00} (this property is known as upper semi-continuity \cite{aubin84}, but if the semi-tangent is uni-valued over $N$, namely $N$ is $C^1$, then this property is equivalent to the continuity of the semi-tangent section).
When this happens we say that the semi-tangent section is {\em regular}.
In conclusion, the normalization via the auxiliary timelike vector field $T$ proves that regular semi-tangent sections exist and that over every generator the induced map $n$  satisfies the pregeodesic equation  $\nabla_n n=\alpha n$ where $\alpha$ is continuous over the generator.

Note that if $N$ were $C^2$ then $n$ could be regarded as  a global $C^1$ field over $N$ and so $\alpha$ would be a continuous functions over $N$. These properties will not be used in the following theorem, although they will hold in applications in the reminder of the work.


We stress that in this section we are not imposing any compactness, totally geodesic, or Killing conditions (as it is clear by the lack of sufficient regularity). With the following theorem we clarify what it means to find a field $n$ of constant surface gravity under very general circumstances. In the following section we shall apply it to the special case in which  $N$ is an event horizon $H^+$ (for which we have established strong regularity properties).

In the next result  ``forward complete''  does not refer to the affine completeness of some geodesic. Rather it refers to the fact that the parameter $s$ such that $n=\frac{\dd }{\dd s}$ over the geodesic is unbounded from above. For a previous  result in the  direction of the following theorem, see \cite[Prop.\ 3.3]{costa24}.\\

\begin{theorem} \label{vm4p}
On a $C^0$  future  null hypersurface $N$ the following two properties
are equivalent for a regular semi-tangent section $n$
\begin{itemize}
\item[(a)] 
there is a  finite constant $\kappa<0$ such that for every
$p\in N$ any generator starting at $p$ with tangent belonging to the semi-tangent section has finite  affine length $\Lambda=-\frac{1}{\kappa}$ in the future direction,
\item[(b)]      there is a  finite constant $\kappa<0$ such that over every generator of $N$
\[
\nabla_n n =\kappa n,
\]
and  $n$ is  forward complete.
\end{itemize}
If they hold then the number $\kappa$ mentioned there is the same for both instances and any other regular semi-tangent section $n'$ with the same properties satisfies  $n'=\frac{\Lambda}{\Lambda'} \, n$.  Namely, fixing the value of $\Lambda>0$ or $\kappa<0$ fixes the field $n$. Conversely, fixing the scale for $n$ fixes $\kappa$ and $\Lambda$. A dual statement also holds.
\end{theorem}

\begin{proof}
Let us focus on a generator and let us denote with $n$ the $C^1$ tangent field over it coming from the regular semi-tangent section. As shown above
$\nabla_{n}  n=\alpha n$, where $\alpha$ is a function.
Let $x:[0,b)\to N$ , $s \mapsto x(s)$,  be the (generator) integral curve of $n$, with $x(0)=p$.
As proved in \cite[Lemma 2]{minguzzi21} (for the case $b=+\infty$ but the proof does not change), the geodesic $\gamma$ starting from $x(\tau)$ with tangent $\dot \gamma=n$ has future affine length
\begin{equation} \label{bud}
\Lambda(x(\tau))=\int_\tau^b e^{\int_\tau^\rho \alpha(x(s))\dd s} \dd \rho.
\end{equation}
In particular,  $1 + \alpha \Lambda + n(\Lambda) = 0$.

(a) $\Rightarrow$ (b). Evidently, the assumption (a) implies that $\Lambda$ is constant over the generator and so that the last term in the left-hand side of $1 + \alpha \Lambda + n(\Lambda) = 0$ vanishes. This implies that $\kappa=\alpha$.
From (\ref{bud}) we get as, $\alpha$ is constant,
\begin{equation} \label{bud2}
\Lambda(x(\tau))=\int_\tau^b e^{\kappa (\rho-\tau)} \dd \rho=\frac{e^{\kappa (b-\tau)} -1}{k}.
\end{equation}
Suppose that $b<+\infty$, then $\Lambda(x(\tau))$ depends on $\tau$  and so $\Lambda$ is not constant. The contradiction proves that $b=+\infty$ and hence that $n$ is forward complete. Since $\alpha=\kappa$ we also get $\nabla_{n}  n=\alpha n=\kappa n$.

(b) $\Rightarrow$ (a). If $b=+\infty$ and $\nabla_{n}  n=\kappa n$ then $\alpha=\kappa$ thus again integrating (\ref{bud}) we get $\Lambda=-1/\kappa$ which does not depend on the starting point on $N$ and on the generator chosen.

Due to (a) the last statement is clear. $\square$
\end{proof}


We arrive at the definition of surface gravity. Note that we are not interested in assigning a number to each point of $N$ (this is known as {\em dynamical surface gravity} among physicists \cite{pielahn11}) but rather a number to $N$.

\begin{definition}[Surface gravity $\kappa$ and its sign] \label{vqog}
We say that a $C^0$ future null hypersurface $N$ has {\em negative surface gravity}  if there exists a semi-tangent section $n$ with the properties of Theorem \ref{vm4p}.
We say that a $C^0$ past null hypersurface $N$ has {\em positive surface gravity} if   Theorem \ref{vm4p} applies  in the time dual version. We say that a $C^1$ null hypersurface has {\em zero surface gravity} if a geodesic field $n$ can be found such that all generators are complete.
When the scale of $n$ can be fixed the {\em surface gravity} $\kappa$ is defined. It is zero if the last case apply, otherwise it can be read from $\nabla_n n=\kappa n$ or $\kappa=-1/\Lambda$.
\end{definition}

Note that   a $C^0$ future null hypersurface is also a $C^0$ past null hypersurface iff it is a $C^1$ null hypersurface. In fact, by local achronality every future generator starting at $p$ must match (no corner) a past generator at $p$, which implies that every point is in the interior of a generator.

Even in the $C^1$ case the affine completeness properties are all mutually excluding so the previous definition is consistent. The formula $\Lambda=-1/\kappa$ holds also in the positive surface gravity case, where a negative $\Lambda$ signals past incompleteness.

Note that if $n$ is geodesic then completeness is equivalent to affine completeness.

In the following section we shall show how the scale of $n$ is fixed for the event horizons.

\begin{remark}
Theorem \ref{vm4p} allows us to recognize the importance of the completeness of the field $n$ for the definition of the sign of surface gravity. It is not sufficient to refer to the pregeodesic equation $\nabla_n n =\kappa n$ with constant $\kappa$ to read the value of the sign surface gravity, for it is necessary to check that $n$ is forward complete for $\kappa <0$ (an inequality which implies future affine incompleteness), backward complete for $k>0$ (an inequality which implies past affine incompleteness), and complete for $\kappa =0$. This check is not needed for compact $N$, as the completeness property would be ensured by standard ODE theory, but is otherwise necessary.

For instance, over a smooth event horizon which admits a section we can find a geodesic $n$. It is sufficient to fix the value of $n$ on the section and then propagate it geodetically. Though in this way we get $\nabla_n n=0$ we cannot conclude that $\kappa=0$ as $n$ might be non-complete.

Therefore, the {\em sign of surface gravity}  is related to two contrasting forms of (in)com\-ple\-te\-nesses (which, actually, in the degenerate case coincide).
The reader can compare our definition with the suggestion by Petersen \cite[Def.\ 1.22]{petersen19} where the  pregeodesic equation is paired instead with a   condition of the form, $[k,n]=0$ on $N$, with $k$ Killing. We shall meet this type of property later on but it will be deduced, see Thm.\ \ref{voz2}.
\end{remark}

\section{The surface gravity (temperature) and angular velocity of event horizons} \label{surf}

The goal of this section is to explore in some detail the geometry of the horizon in 4-dimensional spacetimes under the dominant energy condition. In any case, all assumptions will be spelled out in the theorems below. Some results and strategies can actually work in other dimensions.

Our main result establishes that it is possible to assign a {\em constant surface gravity} (not just the sign), to be interpreted as temperature via the formula $T=\frac{\kappa}{2\pi}$, an {\em angular velocity}, a {\em privileged lightlike vector field} $n$ tangent to the horizon, and a {\em privileged vector field} $\zeta$ on the horizon generating an $S^1$-action (axisymmetry). These fields are related as follows:
\[
k=n+\omega \zeta.
\]
These results are not entirely new, of course, the statement concerning the existence of an axisymmetry being related to what in the literature goes under the name of {\em black hole rigidity theorem}.

It is worth recalling the general traditional strategy for this type of investigation, also to emphasize the differences from our findings. The standard approach, due to Hawking \cite{hawking72,chrusciel97}\cite[Prop.\ 9.3.6]{hawking73}, is to prove under the following conditions:
\begin{itemize}
\item[(a)] analyticity of the metric and horizon,
\item[(b)] (electro-)vacuum assumption or assumptions on the hyperbolic evolution of matter from Cauchy data,
\item[(c)] asymptotic flatness conditions and causality assumptions,
\item[(d)] energy conditions,
\item[(e)] bifurcation conditions,
\end{itemize}
that a stationary spacetime admitting a Killing field $k$ which is timelike near the conformal boundary must have a {\em Killing event horizon} and must be axisymmetric (the Killing field $n$ tangent to the horizon might be different from $k$). The $S^1$-action Killing field $\zeta$ is first found in a neighborhood of the horizon and then extended by using real analyticity. The Killing property of the horizon also implies constant surface gravity, via simple algebraic computations making use of the dominant energy condition \cite{bardeen73,chrusciel20}. Similar rigidity results hold in higher dimensions \cite{hollands07,hollands09}.

Once these results are established, other arguments, not depending on the analyticity assumption, first explored by Carter \cite{carter71b} and Robinson \cite{robinson75}, show that the black hole is determined by just three parameters: mass, angular momentum, and electric charge. This result is known as the {\em no-hair theorem}. This line of inquiry thus leads to the proof of the {\em black hole uniqueness theorem} in the analytic case, which essentially states that the only spacetime with the above (suitable) properties is the Kerr(-Newman) solution \cite{robinson75,mazur82,heusler96,amsel10,chrusciel10b,meinel12,chrusciel12b,chrusciel20}. Incidentally, by our Theorem \ref{bos}, any analytical regularity assumption on the horizon, as mentioned in (a), can be dropped.

Our main contribution in this section goes in a different direction. To start with, we do not impose an electro-vacuum assumption or assumptions on the hyperbolic evolution of matter from Cauchy data. Even asymptotically, we do not require the spacetime to be vacuum, just stationary. We can imagine a stationary configuration in which an ingoing matter-energy flow balances a similar outflow. There can still be an event horizon, in which case our result shows that we can assign a temperature and an angular velocity to it.

We do not impose any bifurcation assumption, thus completely dropping (a), (b), and (e). Of course, our results are correspondingly weaker. Indeed, we stress that the philosophy of our approach is different from Hawking's. We do not try to prove that the horizon is Killing from the outset; in fact, such a problem is difficult in the smooth case \cite{moncrief83,isenberg85,friedrich99,racz00,ionescu09,alexakis10,petersen19}, and has been solved so far only in the non-degenerate case and again under vacuum assumptions (see the discussion at the end of the section; in any case, typically one needs to prove at the very first step that the surface gravity is constant without using the analyticity assumption). Instead, we study the {\em geometry of the horizon in itself} and show the existence of the isometries just on the horizon. The temperature and angular velocity are thus shown to be intrinsic to the horizon.\footnote{In a way, our approach is closer in spirit to the study of near-horizon geometry \cite{kunduri13,dunajski25}, for in that setting one focuses on the intrinsic geometry of the horizon as it follows from the embedding in the spacetime by somehow projecting/reducing the Einstein equations to it (but in those works the horizon is assumed to be Killing and degenerate from the outset, and the vacuum equations are imposed,  all assumptions that we do not impose).}

If we assume an (electro-)vacuum condition, our findings are connected to the attempt to prove the black hole uniqueness theorem in the smooth framework. As is well known, there has been only limited success as analyticity cannot be used to extend various objects obtained at and near the horizon. Suppose such uniqueness is false. Our result shows that, nevertheless, we can still assign a temperature and an angular velocity to the horizon regardless of it being Killing.

Of course, there is another take on our results. Our hope is that they might ultimately help to prove the black hole uniqueness theorem in the smooth setting. The idea is to start constraining the geometry of the horizon, and learning as much as possible about it, as we try to do in this section. From here, PDE techniques could hopefully be used to get the relevant extension over the whole domain of outer communication.

We shall make use of the following result.

\begin{lemma} \label{kilr}
Let $(N,g)$ be a compact Riemannian manifold, $n$ a Killing field such that $g(n,n)=1$, and $k$ a vector field such that $[k,n]=fn$ for some function $f: N
\to \mathbb{R}$ such that $\p_n f=0$. Then $f=0$ and hence $[k,n]=0$.
\end{lemma}

\begin{proof}
We denote with $\nabla$ the Levi-Civita connection of $g$. From   $[k, n] = fn$, we have, taking the scalar product with $n$
 \[
  {g}({\nabla}_kn,n) - {g}({\nabla}_nk,n) = f.
  \]
 As $n$ is a unit Killing vector field for $(N, {g})$, ${\nabla}_nn = 0$, so we get   $- n({g}(k,n)) = f$.
  From   $[k, n] = fn$, we also have, taking the scalar product with $k$
 \[
  {g}({\nabla}_kn,k) - {g}({\nabla}_nk,k) = f{g}(k,n).
  \]
 Since  $n$ is a  Killing vector field for $({N}, {g})$, we get $ n({g}(k,k)) = -2f{g}(k,n).$ From this we get
 $ n( n({g}(k,k))) = -2fn({g}(k,n))$ since $n(f) =0$.
 But $- n({g}(k,n)) = f$, so it follows that
 $  n( n({g}(k,k))) = 2f^2$. Let us set $h := {g}(k,k): N\to \mathbb{R}$ which, being continuous and defined on a compact set, is  bounded. For any integral curve $\gamma:\mathbb{R} \to N$ of $n$, we have $(h\circ\gamma)^{\prime\prime}(t)\geq 0,  \forall\, t \in \mathbb{R}; $ hence $h\circ\gamma$ is a convex function defined on $\mathbb{R}$, but since  $h\circ\gamma$ is bounded, it is constant. It follows that for any integral curve $\gamma$ of $n$, $(h\circ\gamma)^{\prime\prime}(t)=  0, \forall\, t \in \mathbb{R}$; which proves that $f =0$. Hence  $[k, n] = 0$. $\square$
\end{proof}

In the next result the quotient projection along the generators is always denoted $\tilde \pi$, even when it applies to the compactified horizon $\hat H$ rather than $H^+$. We recall that $\phi$ is the flow of the Killing field $k$, while $\varphi$ is the flow of a lightlike field $n$ tangent to the horizon.

For simplicity, we formulate the following two theorems for a connected event horizon $H^+$, though they hold equally well for a connected component of $H^+$, provided it satisfies the properties mentioned in the theorems.

\begin{theorem}[Compactified horizon with closed generators] \label{voz} \\
Let $(M,g)$ be a 4-dimensional  spacetime endowed with a $C^k$, $k\ge 3$, metric. Let $k$ be a complete  Killing vector field.
Let $\Sigma_{end}$ be an acausal  hypersurface, possibly with edge, over which $k$ is timelike.  Let $M_{end}:=\phi(\Sigma_{end})$ and suppose that  the horizon $H^+:=\p I^-(M_{end})\cap I^+(M_{end})$, is connected, has compact Killing projection $S=\pi(H^+)$ (cf.\ Def.\ \ref{bid}) of non-vanishing Euler characteristic,  and that strong causality holds on $H^+$. Furthermore, assume that $H^+$ is as regular as the metric and totally geodesic (or assume $\theta\ge 0$ in support sense over $H^+$ and the null convergence condition so as to apply Theorem \ref{bos}).

Then the constant $\tau>0$ in Cor.\ \ref{bjz} can be chosen so that $\phi_\tau$ sends each generator $\gamma$ to itself (and for any $p\in \gamma$, $\phi_\tau (p)> p$), namely on $\hat H$ the generators close and $\hat \pi: \hat H \to S$ is a trivial $S^1$-principal bundle over the base manifold $S$ whose fibers are closed Killing orbits.

Moreover, $\tilde \pi: \hat H \to \tilde S$ is a trivial $S^1$-principal bundle  whose fibers are closed lightlike geodesics. More precisely, a future-directed lightlike vector field $n$ tangent to the generators can be found such that its flow $\varphi$ generates an $S^1$-action with period $\tau$, $\varphi_{\tau}=Id$ (we say that $n$ is $S^1$-generating).

The action of the Killing field sends the fibers of $\tilde \pi: \hat H \to \tilde S$ to fibers (it is fiberwise).

Let $\check \tau$ be the infimum of the possible $\tau$ with the above properties. It is really a minimum (so  $\check \tau>0$, and we speak of {\em rotational case}) iff $k$ is not tangent everywhere to the generators, in which case any other $\tau$ is a multiple. Otherwise $\check \tau=0$ ({\em irrotational case}) and every $\tau>0$ has the above properties.

The vector field $n$  lifts from $\hat H$ to $H^+$ to a complete vector field (denoted in the same way). The statement on $n$ being $S^1$-generating on $\hat H$ with period $\tau$ becomes on $H^+$, $\varphi_{\tau}=\phi_\tau$, where $\varphi$ and $\phi$ are the flows of $n$ and $k$ respectively.
\end{theorem}

\begin{proof}

The first part of the proof goes precisely as in  the proof of \cite[Prop.\ 2.1]{friedrich99}.
Since $H^+$ is totally geodesic the geodesic flow preserves the metric induced on $H^+$. Thus on the quotient $\tilde S$ we have a well defined Riemannian metric. The Killing vector field $k$ on $(H^+, g \vert_{TH^+\times TH^+})$ (with flow $\phi$) maps generators to generators and descends to a Killing vector field $\tilde k$ on $\tilde S$, (with flow $\tilde\phi$). Observe that $\tilde k$ necessarily vanishes at some point, as the Euler characteristic of $\tilde S$ is different from zero. The next step in the argument is as in pages 119--120 of \cite{wald94}, where it is proved that $\tilde\phi_\tau=Id_{\tilde S}$ for some $\tau>0$ (and a minimal $\tau$ can be chosen if $k$ does not identically project to zero, in which case $k$ and $n$ are proportional on $H^+$). Consequently, $\phi_\tau$ maps each generator to itself.

The statement in parenthesis on the fact that generators are moved to the future by $\phi_\tau$ follows from Thm.\ \ref{vmlh}.

Since $\hat H$ is obtained from $H^+$ by quotient with respect to the $\mathbb{Z}$-action generated by $\phi_\tau$, the Killing flow induced on $\hat H$ generates a $S^1\simeq \mathbb{R}/\mathbb{Z}$-free action.

The projection $\hat \pi$ is constructed by noticing that every sufficiently small neighborhood of $p\in \hat H$ is diffeomorphic with a neighborhood of point in the fiber $c^{-1}(p)$ in the covering space (remember that we have the covering map $c: H^+\to \hat H$). Which point is chosen is irrelevant as different choices shall have canonically diffeomorphic neighborhoods. Thus composing the local diffeomorphism with the projection $\pi$ constructed in Cor.\ \ref{bjs} we get $\hat \pi$ in a neighborhood of $p$. This shows that the bundle $\hat \pi: \hat H\to S$ does indeed exist.

Similarly, composing the local diffeomorphism with $\tilde \pi$ and using Thm.\ \ref{vobn} we get that the bundle $\tilde \pi: \hat H\to \tilde S$ does indeed exist.
Its fibers are circles. Furthermore, by taking a local section $\Sigma$  of $\tilde \pi: H^+\to \tilde S$ passing through $p$, each point of $\Sigma$ is mapped to $\phi_\tau(\Sigma)$, a section passing through  $\phi_\tau(p)$ (each point $q\in \Sigma$ is moved to its causal future, see Thm.\ \ref{vmlh}). As $\Sigma$ and $\phi_\tau(\Sigma)$ get identified, $\hat H$ is covered by products $S^1\times \Sigma$, where $\Sigma$ is diffeomorphic to its projection on $\tilde S$. In other words, $\tilde \pi: \hat H\to \tilde S$ is a fiber bundle with (oriented) $S^1$ fibers. By \cite[Prop.\ 6.15]{morita01} every oriented $S^1$ bundle admits the structure of principal $S^1$ bundle. In other words, there is a tangent field $n$ to the generators such that its $n$-parameter over the closed generator has length $\tau$ (regardless of the generator).


The principal bundle $\tilde \pi: \hat H\to \tilde S$ is  trivial because, if $\tilde s: \tilde S\to H^+$,  is a smooth section for $\tilde \pi:H^+\to \tilde S$,  then $c\circ \tilde s:\tilde S\to \hat H$ is a section for   $\tilde \pi: \hat H\to \tilde S$.

By the compactness of $\hat H$, the vector field $n\in \mathfrak{X}(\hat H)$ is clearly complete. As a consequence, so is its lift (denoted in the same way) $n\in  \mathfrak{X}(H^+)$.  $\square$
\end{proof}

If $k$ is not everywhere tangent to the generators, the canonical choice for $\tau$ to realize the compactified space $\hat H$ is the minimal possible value $\check \tau$. All the physically relevant quantities will be related to this choice and might also be denoted with a `check'.  We call the constant $\check \tau$ the {\em rotational period of the horizon} ad the constant  $\check \omega:=2\pi/\check \tau$ the {\em angular velocity of the horizon}. The latter is related to a $S^1$-action on $H^+$ as  clarified by the following theorem.


Working with $\tau$ might be convenient to avoid dealing with the cases $\check \tau>0$ and $\check \tau=0$ separately.

Note that $n$ is complete in $H^+$, thus in $H$ it cannot be extended to a vector field unless it converges to zero on $\textrm{edge}(H^+)$.

The proof of the following result uses the compactification of Theorem \ref{voz}, however, for better readability, the statement does not mention such construction.

\begin{theorem}[Lightlike field, surface gravity and angular velocity of the horizon] \label{voz2}

Assume the hypothesis of Theorem \ref{voz} and suppose, furthermore, that the dominant energy condition holds.

On the horizon $H^+$ there is one and only one future-directed $C^{k-1}$ lightlike field $n$, called the {\em lightlike field of the horizon}, with the properties:
\begin{itemize}
\item[(a)] completeness,
\item[(b)] pregeodesic equation
\begin{equation} \label{jjwe}
\nabla_n n =\kappa n,
\end{equation}
 for some constant $\kappa$,
\item[(c)] $\varphi_\tau=\phi_\tau$ ($S^1$-generating condition) for some $\tau>0$.
\end{itemize}
The constant $\kappa$  for the lightlike field of the horizon is thus uniquely determined and called the {\em surface gravity of the horizon}.

For $\kappa=0$ all generators are (affine) complete. For $\kappa>0$ they are all future affine complete, and past affine incomplete in the portion contained in $H^+$ (and time dually for $\kappa<0$).

Further, there are two cases. Either $n=k$ in which case (c) holds for any $\tau$, or $n\ne k$ somewhere on $H^+$ in which case (c) holds for $n$ for a minimal value of $\tau$ denoted $\check \tau>0$. In the former case we set $\check \omega:=0$ and in the latter case we set $\check \omega:=2\pi/\check\tau$, and call $\check \omega$ the {\em angular velocity of the horizon}.

Consider a lightlike geodesic $\gamma: I \to H^+$ running on the horizon starting with velocity $n$. Then we have the vector ratio $\dot \gamma(\check\tau)/d\phi_{\check \tau}(\dot \gamma(0))=e^{-\check c}$ where the constant is $\check c=\kappa \check \tau$ and so is independent of the starting point (note also that $d\phi_{\check \tau}(n)=d\varphi_{\check \tau}(n)=n$). We call $e^{-\check c}$ {\em the dilating factor of the horizon}.

We have on $H^+$
\begin{equation} \label{com7}
[k,n]=0,
\end{equation}
and if  $\check \tau >0$ the field $\zeta:=\frac{\check \tau}{2 \pi} (k- n)$ provides  a $S^1$-action on $H^+$ with period $2\pi$ (the axisymmetry action; we can also set $\zeta:=0$ for $\check \tau=0$). Its non-trivial orbits are spacelike.
The  flows of the fields $k,n,\zeta$  are isometric in the sense that they preserve the  metric induced on the horizon.

There exists a smooth function $\bar f: H^+  \to \mathbb{R}$ such that $\dd \bar f(n)=1$, $\dd \bar f (\zeta)=0$, $\psi^*_\theta\bar f=\bar f$, where $\psi$ is the flow of $\zeta$. Its level sets are  smooth  cross-sections of $H^+$, intersected  exactly once by both generators and Killing orbits, and   the orbits of $\zeta$ are tangent to  the level sets.
The level sets are sent to level sets by the action of the Killing flow $\phi$ and by the generator flow $\varphi$.
\end{theorem}

The proof is somewhat long so we organize it in sections  with titles in boldface.

\begin{proof}
We use the compactification $\hat H$ relative to the parameter $\tau>0$ introduced in Theorem \ref{voz}.

{\bf The dilating factor and its  generator independence}.
Let $n$ be a $S^1$-generating field with  period $\tau$ on $\hat H$. Since the bundle $\tilde \pi: \hat H\to \tilde S$ is trivial we can introduce a section $\Sigma$ and push it with the $n$-flow $\varphi$  to get a foliation. Call the tangents spaces to this foliation {\em horizontal}. They are preserved by the flow $\varphi$ of $n$. We can introduce a 1-form field $n^*$ such that $n^*(n)=1$ and its kernel is horizontal. At this point we proceed with the ribbon argument, that is,
let us consider two integral curves $x_0(s)$ and $x_1(s)$ of $n$ starting from $\Sigma$ and terminating in another leaf $\Sigma'$. We can connect the starting points with a horizontal curve  in $\Sigma$ and the ending points with another horizontal curve  in $\Sigma'$ which is the image under the flow of the former (note that since $H^+$ is connected, $\tilde S$ is connected), and we have by \cite[Eq.\ (18)]{minguzzi21} (which uses the dominant energy condition), by choosing the longitudinal length to be that of $m$ cycles
\[
 \vert \int_0^{m\tau}\kappa(x_1(s)) \dd s -\int_0^{m\tau} \kappa(x_0(s))\dd s\vert \le 2B ,
\]
where the constant $B$ does not depend on how much elongated is the ribbon, that is, it does not depend on $m>0$. Here $s$ is the parameter such that $n=\frac{\dd}{\dd s}$. However, in the present specific case $\kappa(x_1(s))$, $\kappa(x_0(s))$ are periodic with period $\tau$ thus the inequality is satisfied if and only if
\[
\int_0^{\tau}\kappa(x_1(s)) \dd s=\int_0^{\tau}\kappa(x_0(s)) \dd s .
\]
As  $\kappa(x(s))$ is periodic with period $\tau$, we conclude that the integral
\[
c:=\int_0^{\tau}\kappa(x_p(s)) \dd s,
\]
with $x_p(s)$ integral curve of $n$ with starting point $p$, is actually independent of $p\in \hat H$.

 This integral has the following interpretation. Let us consider the affinely parametrized geodesic $\gamma(t)$ with initial conditions $\gamma(0)=x(0)$ and $\dot \gamma(0)=b n(x(0))$ for some $b>0$. We have $\dot  \gamma(t(s))=f(s) b n(x(s))$, and from the geodesic condition $\nabla_{\dot \gamma}\dot \gamma=0$ it follows \cite[proof of Lemma 2]{minguzzi21}
 \[
 f(s)=\exp[-\int_0^s \kappa(x(r)) \dd r],
  \]
  thus after a cycle the tangent vector gets expanded by $e^{-c}$.  If $c< 0$ we have future incompleteness and past completeness for all generators \cite[Comment after Prop.\ 6.4.3]{hawking73}, if $c>0$ we have future completeness and past incompleteness for all generators, and if $c=0$ we have completeness in both directions for all generators. If $c\ne 0$ we speak of non-degenerate case and if $c=0$ of degenerate case.
The constant $c$ is denoted $\check c$ when we choose $\tau=\check \tau$, see Thm.\ \ref{voz} for the definition of $\check \tau$.

 Notice that  this property of $c$ is formulated using geodesics and so  does not mention the $S^1$-generating property of $n$. We just used the closure of the generators on $\hat H$,  but no special parametrization.

  Once $n$ is lifted to $H$ the just proved property becomes that expressed as the vector ratio $\dot \gamma(\check\tau)/d\phi_{\check \tau}(\dot \gamma(0))=e^{-\check c}$ in the statement of the theorem, where the choice $\tau=\check \tau$ has been made, indeed, note that $\dot \gamma(0)=b n(\gamma(0))$ thus $d\phi_{\check \tau}(\dot \gamma(0))=d\varphi_{\check \tau}(\dot \gamma(0))=b d\varphi_{\check \tau}(n)=b n(\gamma(\check \tau))$ and so the ratio in the statement of the theorem is precisely that we calculated above.

{\bf Existence of $\bm{S^1}$-generating fields on $\bm{\hat H}$ of constant surface gravity}.
Let us consider the non-degenerate case, $c \ne 0$, namely suppose that there is a (past or future) incomplete generator on $\hat H$.
By the results in \cite{reiris21,minguzzi21} we know that $n$ can be chosen so that it has non-zero constant surface gravity. Actually, in this simplified setting of trivial  $S^1$-bundle this result can be obtained directly by specializing the formulas obtained in \cite{minguzzi21}. Let $n$ be a $S^1$-generating field with  period $\tau$. Recalling that  $\hat H$ is diffeomorphic to $S^1\times \Sigma$, we  define the function on $\hat H$
\begin{align}
e^{f(x_q(r))}&=\frac{c}{\tau(e^c-1)} \int_0^\tau \left[ \exp \int_0^t\kappa(x_q(r+s)) \dd s\right] \dd t\\
&=\frac{c}{\tau(e^c-1)} \int_r^{r+\tau} \left[ \exp \int_r^u\kappa(x_q(s)) \dd s\right] \dd u
\end{align}
where $x_q$ is the integral curve of $n$ starting from $q\in \Sigma$. The first expression shows that the function is periodic $f(x_q(r))=f(x_q(r+\tau))$ and so well defined. The second expression shows that the surface gravity of $n'=e^f n$ is $\kappa'=e^f(\p_n f + \kappa)=c/\tau$ which is a non-zero constant.


So let $n$ be any tangent field with constant non-zero surface gravity on $\hat H$. We can still calculate the dilating factor as done previously to get
\[
e^{-c}=\exp[-\int_0^{R(\gamma)} \kappa(x(r)) \dd r]=\exp (-\kappa R(\gamma))
\]
where this time the range of the parameter might depend on the generator as $n$ is not $S^1$-generating a priori. But the just proved equation proves that $R(\gamma)=c/\kappa$, so it is indeed independent of $\gamma$. By rescaling $n$ with a global constant we can accomplish $R=\tau$.



We have just proved that, in the non-degenerate case, $c\ne 0$, there is a choice of $n$  on $\hat H$ of constant non-zero surface gravity   and for any such choice $n$ is automatically $S^1$-generating.



In the degenerate case, $c=0$, it is similarly possible to construct a  $S^1$-generating field of zero surface gravity.
 Let us define the field over a smooth section and propagate it with the geodesic property. This extends to a well defined field over $\hat H$ since $e^{-c}=1$. The new field has zero surface gravity. But we have still the freedom of choosing the initial field. Let $\lambda(\gamma)$ be the affine length of the first cycle, which depends on the generator chosen.  Rescaling $n\to  n\lambda(\gamma)/\tau$ over each generator $\gamma$  we get  that the new $n$ has zero surface gravity and is $S^1$-generating with period $\tau$ (one could also start from a $S^1$-generating field with period $\tau$ and use the previous formulas for $e^f$ in the limit $c\to 0$).



The field $n$ with the required properties  (a)-(c) is obtained lifting the analogous field from $\hat H$.
For  $\check \tau =0$ we just  consider the construction of $n$ on $\hat H$ of the previous paragraphs choosing any $\tau>0$ for the compactification. For  $\check \tau >0$ we just  consider the construction of $n$ on $\hat H$ of the previous paragraphs  choosing  $\tau=\check\tau$ for the compactification.

{\bf Uniqueness of $\bm{n}$ and $\bm{\kappa}$ on $\bm{H^+}$}.


Assume that the future-directed lightlike field $n$ satisfies (a)-(c).

Suppose $\check \tau=0$. We already know from Thm.\ \ref{voz} that $n=\alpha k$ for some function $\alpha$  which means that the horizon is Killing, and since both fields do not vanish on $H^+$, using Thm.\ \ref{vmlh}, $\alpha > 0$.
Since the dominant energy condition  holds we know that $\nabla_k k =\kappa' k$ for some constant $\kappa'$, see \cite[Thm.\ 4.3.12]{chrusciel20}. Thus both $n$ and $k$ have constant surface gravity. But the sign of surface gravity must be the same as it is related to the affine completeness of the generators. If the surface gravity does not vanish, as both $n$ and $k$ are complete on $H^+$,  Theorem \ref{vm4p} implies that they differ by a constant. Condition (c) now implies $k=n$. If the surface gravity does vanish
 then $0=\alpha^{-1} \nabla_n n=\nabla_k n=(\p_k \alpha) k + \alpha \nabla_k k =(\p_k \alpha) k$ which implies that $\alpha$ is a constant over every generator. But this constant must again be 1 due to condition (c). We conclude that for $\check \tau=0$, $n=k$, thus $n$ is unique.

Suppose $\check \tau>0$. Property (c) implies that all geodesics return to themselves under $\phi_\tau$, thus $\tau=i \check \tau$ by Thm.\ \ref{voz}, where $i$ is a positive integer.  Let $n'$ be a $S^1$-generating field of constant surface gravity which we know to exist on  $\hat H_{\check \tau}$ (obvious meaning of the notation). Its lift to $H^+$ is denoted in the same way. Both $n$ and $n'$ on $H^+$ are complete and have constant surface gravity. The signs of the surface gravities actually coincide because they are related to the affine incompleteness of the generators. If the surface gravities do not vanish then Theorem \ref{vm4p} proves that $n$ and $n'$ differ by a constant factor. As both satisfy $\varphi_\tau=\phi_\tau$ this factor is unity. If the surface gravity vanish then, writing $n=\alpha n'$, $\alpha >0$, $0=\alpha^{-1} \nabla_n n=\nabla_{n'} n=(\p_{n'} \alpha) n' + \alpha \nabla_{n'} n'=(\p_{n'} \alpha) n'$ which implies that $\alpha$ is a constant over every generator. But this constant must again be 1 due to condition $\varphi_\tau=\phi_\tau$, which proves $n=n'$ and so all possible $n$ coincide. In particular, the identity $n=n'$ implies $\varphi_{\check\tau}=\phi_{\check \tau}$.

{\bf Commutativity $\bm{[k,n]=0}$}.
Since the flow of $k$  on $\hat H$ sends generators to generators, $k$ is projectable under the map $\tilde \pi: \hat H \to \tilde S$. But $\tilde \pi_*(n)=0$, thus $0=[\tilde \pi_*(k),\tilde \pi_*(n)]=\tilde \pi_*([k,n])$ which implies $[k,n]= f n$ for some function $f: \hat H \to \mathbb{R}$. But since $L_k \nabla =0$ we have, taking the Lie derivative of $\nabla_n n = \kappa n$, $\nabla_{fn} n+\nabla_n (fn)=\kappa f n$ which implies $\p_n f +\kappa f=0$ and hence along  a generator $\gamma(s)$, $f=C \exp(-\kappa s)$ where $s$ is a parameter such that $n=\dd/\dd s$. But since the generator closes itself after a period $\check \tau$, $f(\gamma(0))=f(\gamma(\check\tau))$ which implies $C =C \exp(-\kappa \check\tau)$.

Suppose $\kappa \ne 0$, then the only solution is $C=0$, that is $f=0$  over every generator, and hence $f=0$ over $\hat H$.

As for the case $\kappa=0$ we proceed as follows.
We recall that the bundle $\tilde \pi: \hat H \to \tilde S$ is trivial, thus we can find a 1-form $\omega$ such that $\omega(n)=1$, $L_n\omega=0$ (in the trivialization $\hat H \to S^1_\theta\times \tilde S$, let $\omega=\dd \theta$ where $n=\frac{\p}{\p \theta}$), so introduced a Riemannian metric $\sigma$ on $\tilde S$,  $n$ is Killing for the Riemannian metric $\hat g:=\tilde \pi^*\sigma+\omega\otimes \omega$ and normalized, $\hat g(n,n)=1$.  Recalling the equation just obtained, $n(f) = -\kappa f=0$, thus the conditions of Lemma \ref{kilr} apply for $(\hat H, \hat g)$, which implies $[k,n]=0$.

%



{\bf The $\bm{S^1}$-action on $\bm{H^+}$}. The statement on the axisymmetry follows from the formula for the flow of a  linear combination of vector fields that commute   $\Phi_{\frac{\check \tau}{2\pi}(k-n)t}=\Phi_{\frac{\check \tau}{2\pi}kt}\circ \Phi_{-\frac{\check \tau}{2\pi} nt}=\phi_{\frac{\check\tau}{2\pi} t} \circ \varphi^{-1}_{\frac{\check \tau}{2\pi}t}$ which for $t=2\pi$ gives the identity.

Since $k$ is Killing and $H^+$ is totally geodesic, it is clear that its flow preserves the induced metric. It is also well known that the flow of  $n$  preserves the induced metric (e.g.\ \cite[Lemma B1]{friedrich99}\cite{moncrief08}\cite[Lemma 7]{minguzzi21}). Thus the flow of $\zeta$ preserves the induced metric. This means that over every non-trivial orbit (which is closed) $\zeta$ has the same square $g(\zeta,\zeta)$ at each of its points and hence that $\zeta$ has the same causal character at every point of the orbit. It cannot be lightlike as there would be a closed causal curve in contradiction with the strong casuality of $H^+$,  and it cannot be timelike as it is tangent to the horizon, thus $\zeta$ is spacelike wherever it is non-zero.

{\bf Existence of nice cross-sections}. For the last statement, we already know by Thm.\ \ref{voz} that the bundle  $\tilde \pi: H^+ \to \tilde S$ is trivial and admits a smooth section.

We can thus find a surjective function $f: H^+  \to \mathbb{R}$ (essentially the first projection in the trivialization) such that $\dd f(n)=1$, and the level sets of $f$ are images of smooth sections for the bundle.  If $\check \tau=0$  we have $n=k$, $\zeta=0$, and we can set $\bar f:=f$.
For $\check \tau>0$, let $\psi_s$ be the flow of $\zeta$. As $\psi$ commutes with $\varphi$ it sends (parametrized) generators to (parametrized)  generators, so that $(\psi_s)_* (n)=n$.  In particular, for any fixed $s$,  $f \circ \psi_s$ increases over the generators.

Let us consider the function $\bar f: H^+  \to \mathbb{R}$
\begin{align}
\bar f:= \frac{1}{2\pi}\int_0^{2\pi} \psi_\theta^* f \, \dd \theta.
\end{align}
It is smooth and such that $\dd \bar f(n)=1$ which, by the completeness of $n$ on $H^+$, implies that it is surjective over every generator.  This implies that the levels sets of $\bar f$ intersect the generators exactly once. Each of the level sets is a cross-section. Since $\phi_\tau=\varphi_\tau$, $f(\phi_{m\tau})=f(\varphi_{m\tau})$ for every integer $m$ which, by continuity, shows that $f$ is surjective over every $k$-Killing orbit.
Furthermore, due to reparametrization invariance  $\psi_\theta^* \bar f=\bar f$, we have $\dd \bar f(\zeta)=0$, which shows that the orbits of the $S^1$ action $\psi$ run over the level sets of $\bar f$, and due to    the expression for $\zeta$, $k=n+\frac{2\pi}{\tau} \zeta$, we have $\dd \bar f(k)=1$, which proves that $\bar f$ is increasing over every $k$-Killing orbit, and so  that the levels sets of $\bar f$ intersect the $k$-Killing orbits exactly once.


The fact that the level sets are preserved by the $\varphi$ flow follows from $\dd \bar f(n)=1$. The fact that they are preserved by the $\phi$ flow follows from $\dd \bar f(k)=1$.
$\square$
\end{proof}

The Killing property of a field is preserved by multiplication by a global constant. Still, we cannot assign this freedom to the Killing field $k$: for instance, if the spacetime approximates Minkowski at infinity, one typically demands that $k$ should represent the velocity field of an observer in Minkowski, i.e.\ it should approach a  normalized field $g(k,k)\to -1$. This shows that, physically speaking, $k$ cannot be rescaled. It follows that, as show in Theorem \ref{voz}, $\check \tau$ and $\check \omega$ are completely fixed by the geometry. As shown by Theorem \ref{voz2}, $\kappa$ and the field $n$ are also uniquely determined.


This important geometrical fact  confirms the  possibility of interpreting the surface gravity as a physical quantity, i.e.\ the temperature of the black hole $T=\frac{\kappa}{2\pi}$.

Thus the previous theorem clarifies that the geometry fixes the value of surface gravity not just its sign. If surface gravity were introduced via the property (\ref{jjwe}) or with (\ref{jjwe}) and (\ref{com7}) as done e.g.\ in \cite[Def.\ 1.22]{petersen19} only the sign of $\kappa$ would be definable as any point independent rescaling of $n$ would be allowed. It is the condition $\phi_\tau=\varphi_\tau$, and hence the request that $n$ should be $S^1$-generating, that plays a major role in determining the  value of $\kappa$.


Notice that we are able to prove that stationarity implies the existence of an $S^1$ action on the horizon, which we  rightfully term {\em axisymmetry}, without extending the involved fields to Killing fields (in fact our proof is different from previous proofs). It seems remarkable that we can obtain the interesting physical quantities of (constant) temperature and angular velocity directly from the geometry of the horizon without using a vacuum assumption  or any extension to Killing fields (a result which requires the Einstein's equations). Our proof of axisymmetry is also largely independent of dimensionality assumptions.


%
%

As mentioned in  the above theorem, under the dominant energy condition the totally geodesic compact smooth horizon $\hat H$ can either be degenerate ($\kappa=0$) or non-degenerate ($\kappa\ne 0$), namely admit a complete generator or not. This can be seen as a special case of the dichotomy proved in \cite{reiris21,minguzzi21}.


Further, by recent results by Petersen and R\'acz \cite{petersen18,petersen18b,petersen19} in the vacuum  non-degenerate case, the tangent field $n$ that realizes the constant surface gravity can be extended to a Killing vector field  (possibly different from $k$). By lifting it to the neighborhood $V$ of $H^+$ one gets another proof of  Hawking's local rigidity theorem in the non analytic case, which was recently proven by Alexakis-Ionescu-Klainerman \cite{alexakis10} using a different analysis taking advantage of a bifurcate horizon assumption. The mentioned compactification strategy has instead been followed by Petersen \cite{petersen18} (who relied on Chru\'sciel-Costa \cite{chrusciel08} for regularity, so he assumed existence of cross-sections).

We note that in 4 spacetime dimension and for horizons having a section of non-zero Euler characteristic, due to Thms.\ \ref{voz}-\ref{voz2}, there is no need to work in the most general non-closed generators case as done in  \cite{petersen18}. Our result proves the condition $[k,n]=0$ on $H^+$ that Petersen  \cite[Thm.\ 1.23]{petersen19} had to assume in his work and which was not established in \cite{reiris21,minguzzi21}. Thanks to this proof we fill the gap between previous results on the constancy of surface gravity (that do not mention $[k,n]=0$) and the assumption in Petersen's and Petersen and Racz's work.


\section{Conclusions}

Previous results on the smoothness of stationary black holes event horizons relied on a cross-section assumption which implies a $C^1$ differentiability assumption. Through a geometrical analysis we proved, under fairly general causality conditions, the existence of certain fiber bundles in a neighborhood of the horizon, which allowed us to compactify the space and hence apply smoothness result for compact horizons. As a result, we proved smoothness without relying on strong assumptions on the existence of sections (in the case of compact horizons there was a similar development in the literature: the first results came with assumptions on the existence of sections that were only later removed). In fact, our compact projection assumption refers, as the name suggests, to the compactness of a certain projection of the horizon, not to a real `section' i.e.\ a codimension one hypersurface intersecting transversally the horizon.

Ultimately, by using the smoothness of the horizon, we have been able to prove the existence of cross-sections, but the logical derivation was reversed. Finally, we obtained a theorem that describes in detail the geometry of the black hole horizon. We showed the possibility of identifying the surface gravity itself (not just its sign), something that allows its identification with the temperature, and we also identified the angular velocity proving the existence of a horizon axisymmetry. We showed that these quantities are well defined even without imposing the Einstein's equation, just on the basis of the horizon geometry.

It is likely that our methods could be applied, suitably adapted, to the study of higher dimensional black holes and certainly to studies in the analytic setting \cite{hollands07,hollands09}. We leave these research directions for future work.

Summarizing, stressing the physically relevant conclusions, we  showed the constancy of surface gravity and so proved the zeroth's law of black hole thermodynamics without  (a) smoothness assumptions on  the horizon, (b) non-degeneracy conditions, (c) hypothesis on the existence of cross-sections, (d) (electro-)vacuum conditions (e) horizon bifurcation assumptions, and under fairly weak causality conditions (strong causality). Only stationarity was substantially used but this is necessary if one wants to derive a constant temperature, as it expresses the fact that the black hole has reached thermal equilibrium.

\section*{Acknowledgments}
R.H. was supported by ICTP-INDAM ``Research in Pairs'' grant and by ``fondi d'internazionalizzazione''  of the Department of Mathematics of Universit\`a Degli Studi di Firenze thanks to which he visited Florence in August-November 2021, during the covid pandemic.


\section*{Declarations}

\subsection*{Data availability statement}
Data sharing not applicable to this article as no datasets were generated or analysed during the current study.

\subsection*{Conflicts of Interest}
The authors have no relevant financial or non-financial interests to disclose.


\end{document}